\pdfoutput=1
\RequirePackage{ifpdf}
\ifpdf 
\documentclass[pdftex]{sigma}
\else
\documentclass{sigma}
\fi

\usepackage[all]{xy}
\usepackage{bm}
\usepackage{dcolumn}
\usepackage{tikz-cd}

 \newcolumntype{2}{D{.}{}{2.0}}
 \xyoption{2cell}

\newcommand{\hiddenpower}[2] { \ifnum \numexpr#2=1 #1 \else #1^#2 \fi }
\numberwithin{equation}{section}

\def\be{\begin{equation}}
\def\ee{\end{equation}}
\def\ba{\begin{eqnarray}}
\def\ea{\end{eqnarray}}

\usepackage{xparse}
\ExplSyntaxOn
\newcounter{diff_order}
\newcounter{diff_power}

\newcommand{\rawdiff}[3]
{
	\setcounter{diff_order}{0}
	\clist_map_inline:nn{#3}{\stepcounter{diff_order}}
	
	\frac{\hiddenpower{#1}{\thediff_order} #2}
	{
		\def\old_var{DefaultValue}
		\setcounter{diff_power}{0}
		
		\clist_map_inline:nn{#3}
		{
			\def\new_var{##1}
			\ifnum \thediff_power=0
				\stepcounter{diff_power}
			\else
				\tl_if_eq:NNTF \new_var \old_var
				{\stepcounter{diff_power}}
				{
					#1 \hiddenpower{\old_var}{\thediff_power}
					\setcounter{diff_power}{1}
				}
			\fi

			\def\old_var{##1}
		}
		
		#1 \hiddenpower{\old_var}{\thediff_power}
	}
}

\def\Label#1{\label{#1}\ifmmode\llap{[#1] }\else
 \marginpar{\smash{\hbox{\tiny [#1]}}}\fi}
 \def\Label{\label}

\ExplSyntaxOff

\newcommand{\lb}{\left(}
\newcommand{\rb}{\right)}

\newtheorem{thm}{Theorem}[section]
\newtheorem{lem}[thm]{Lemma}
\newtheorem{cor}[thm]{Corollary}
\newtheorem{pro}[thm]{Proposition}

\theoremstyle{definition}
\newtheorem{defn}[thm]{Definition}

\newtheorem{rem}[thm]{Remark}
\newtheorem{exam}[thm]{Example}

\newcommand{\id}{\operatorname{id}}

\DeclareMathOperator{\EEnd}{End}

\renewcommand{\ln}[1]{\text{ln} \lb #1 \rb}

\begin{document}

\allowdisplaybreaks

\newcommand{\arXivNumber}{2211.00451}

\renewcommand{\PaperNumber}{105}

\FirstPageHeading

\ShortArticleName{Quantum Groups, Discrete Magnus Expansion, Pre-Lie and Tridendriform Algebras}

\ArticleName{Quantum Groups, Discrete Magnus Expansion,\\ Pre-Lie and Tridendriform Algebras}

\Author{Anastasia DOIKOU~$^{\rm ab}$}

\AuthorNameForHeading{A.~Doikou}

\Address{$^{\rm a)}$~Department of Mathematics, Heriot-Watt University,
Edinburgh EH14 4AS, UK}
\EmailD{\mail{a.doikou@hw.ac.uk}}
\URLaddressD{\url{https://sites.google.com/view/anastasiadoikou}}

\Address{$^{\rm b)}$~Maxwell Institute for Mathematical Sciences, Edinburgh EH8 9BT, UK}

\ArticleDates{Received July 09, 2024, in final form December 01, 2025; Published online December 11, 2025}

\Abstract{We review the discrete evolution problem and the corresponding solution as a~discrete Dyson series in order to rigorously derive a generalized discrete version of the Magnus expansion. We also systematically derive the discrete analogue of the pre-Lie Magnus expansion and express the elements of the discrete Dyson series in terms of a tridendriform algebra binary operation. In the generic discrete case, extra significant terms that are absent in the continuous or the linear discrete case appear in both Dyson and Magnus expansions. Based on the rigorous discrete derivation key links between quantum algebras, tridendriform and pre-Lie algebras are then established. This is achieved by examining tensor realizations of quantum groups, such as the Yangian. We show that these realizations can be expressed in terms of tridendriform and pre-Lie algebras. The continuous limit as expected provides the corresponding non-local charges of the Yangian as members of the pre-Lie Magnus expansion.}


\Classification{16T05; 17D99; 81R50}

\section{Introduction}\label{section1}

In the present study, we identify interesting links between quantum groups~\cite{Drinfeld, Jimbo, FadTakRes}, and tridendriform~\cite{Dendri, Loday} and pre-Lie algebras (also studied under the name chronological algebras)~\mbox{\cite{Chrono, PreLie, Pre1}} (see also~\cite{Bai, Manchon} for recent reviews). Specifically, we systematically derive the discrete analogues of Dyson series~\cite{Dyson} and Magnus expansion~\cite{Magnus} as solutions of a discrete evolution problem. We then express the discrete Dyson series in terms of a tridendriform algebra, taking into consideration certain extra terms that are not present in the continuous or the linear discrete case and we then precisely derive the discrete pre-Lie Magnus expansion, in analogy to the continuous or the linear discrete case (see~\cite{EbraMan, new3, new4}). We note that the discrete Magnus expansion, albeit a linearized version, was derived for the first time in~\cite{EbraMan}. The discrete version proposed here is an improved generalization in the sense that non-trivial higher-order terms are included, due to choice of the linear operator of the discrete evolution problem. A more precise explanation is provided in Section~\ref{section3}
 (see, e.g., Remark~\ref{remark12} and comments before equation~(\ref{discrete})). As in~\cite{EbraMan}, the use of Rota--Baxter operators~\cite{Rota, Rota2, Rota3} has been essential to express the discrete series in connection with tridendriform and pre-Lie algebras. On the other hand, tensor realizations of quantum groups~\cite{ Drinfeld, Drinfeld1, Drinfeld2, FadTakRes}, such as the Yangian~\cite{Chari, Drinfeld, Drinfeld1, Drinfeld2, Molev}, are also solutions of a discrete evolution problem. Hence, we deduce that the coproducts of the elements of the Yangian can be re-expressed in terms of suitable tridendriform and pre-Lie algebras and this is one of the most important findings of this investigation.

Before we describe in detail what is achieved in each section, we first recall the general setup and some necessary preliminaries on the Magnus expansion as a solution of a linear evolution problem, whereas in the subsequent section we briefly recall basic notions on Rota--Baxter, pre-Lie and tridendriform algebras.
We note that interesting links between pre-Lie algebras, rooted tree graphs, (tri)dendriform and Rota--Baxter algebras have been reported (see, for example,~\mbox{\cite{RBLoday, Ebra2, EbraMan, Dendri}}), whereas links between tridendriform, Rota--Baxter and (quasi)shuffle algebras have been also revealed in~\cite{LodaySh}. Recent findings on the relationships between pre-Lie algebras and braces (nilpotent rings)~\cite{Rump, Smok1, Smok2} have generated increased interest on these distinct algebraic structures opening up unexplored research avenues.
It is worth pointing out that the notion of infinitesimal Hopf algebras and its connections to pre-Lie and dendriform algebras has been explored in~\cite{Aguiar, Aguiar1}. In this study, however, we establish links with typical Hopf algebras, such as the Yangian that also appear in quantum integrable systems.
Some of these profound emerging links within the framework of classical and quantum integrability
{are central themes in our analysis}, while others will be further examined and extended in future works.

We start our discussion by recalling the
initial value problem associated with a linear differential equation. Indeed,
let $A$, $T$ be in general some linear operators (for instance, we focus in this section on $T, A \in \EEnd\bigl({\mathbb C}^{\mathcal N}\bigr)$, ${\mathbb C}^{\mathcal N}$ is the ${\mathcal N} \in {\mathbb Z}^+$ dimensional complex vector space) depending on two parameters, $\xi \in {\mathbb R}$, $\alpha \in {\mathbb C}$, such that
\begin{equation}
\partial_{\xi} T(\xi, \alpha) = \alpha A(\xi) T(\xi, \alpha), \qquad T(x_0,a) =T_0. \label{evo1}
\end{equation}
The formal solution of the evolution equation above can be given as
(we consider simple initial conditions $T(x_0, \alpha) =1$)
\begin{equation}
T(x, \alpha
) = \overset{\curvearrowleft} {\mathrm P} \exp \biggl(\alpha\int_{x_0}^{x}A(\xi) \,{\rm d} \xi\biggr), \qquad x>x_0. \label{eq:CTS_T}
\end{equation}
The latter solution is a {\it path ordered} exponential (called monodromy), which is formally
expressed in terms of Dyson series~\cite{Dyson}
\begin{gather}
\overset{\curvearrowleft} {\mathrm P} \exp\biggl( \alpha \int_{x_0}^{x} A(\xi) \,{\rm d} \xi \biggr)
=
\sum_{n=0}^{\infty}\alpha^n \int_{x_0}^x {\rm d}x_n A(x_n)
\int_{x_0}^{x_n }{\rm d}x_{n-1}A(x_{n-1}) \cdots \int_{x_0}^{x_2} {\rm d}x_{1} A(x_{1}). \label{12b}
\end{gather}

Magnus~\cite{Magnus} suggested that the solution $T(x, \alpha)$ of the linear evolution problem can
be expressed as a real exponential, i.e., \smash{$T(x, \alpha) = {\rm e}^{{\mathcal Q}(x,\alpha)}$} such that
\smash{${\rm e}^{{\mathcal Q}(x,\alpha)}: =1+\sum_{n=1}^{\infty} \frac{{\mathcal Q}^n(x,a)}{n!}$},
where the following formal series are considered
\begin{gather*}
T(x, \alpha) = 1+ \sum_{n =1}^{\infty} \alpha^n T^{(n)}(x), \qquad
{\mathcal Q}(x, \alpha) = \sum_{n=1}^{\infty}\alpha^n {\mathcal Q}^{(n)}(x), \qquad \text{and}
\\
T^{(n)}(x) = \int_{x_0}^x {\rm d}x_n A(x_n)
\int_{x_0}^{x_n }{\rm d}x_{n-1}A(x_{n-1}) \cdots \int_{x_0}^{x_2} {\rm d}x_{1} A(x_{1}). 
\end{gather*}
Comparing the $\alpha$ series expansion of $T(x, \alpha)$ and \smash{${\rm e}^{{\mathcal Q}(x, \alpha)}$}, we obtain the coefficients
\smash{${\mathcal Q}^{(n)}$} as symmetric polynomials of $T^{(n)}$ ($T^{(n)}$, \smash{${\mathcal Q}^{(n)}$} below depend on $x$)
\begin{gather}
 {\mathcal Q}^{(1)} = T^{(1)}, \qquad {\mathcal Q}^{(2)} = T^{(2)} -\frac{1}{2} \bigl(T^{(1)}\bigr)^2,\label{a00} \\
 {\mathcal Q}^{(3)} = T^{(3)} - \frac{1} {2} \bigl( T^{(1)} T^{(2)} + T^{(2)} T^{(1)}\bigr) +\frac{1}{3} \bigl(T^{(1)}\bigr)^3, \qquad\dots, \label{ac}
\end{gather}
and {vice-versa} all $T^{(m)}$'s can be expressed in terms of ${\mathcal Q}^{(m)}$'s,
\begin{gather}
T^{(1)} = {\mathcal Q}^{(1)} , \qquad T^{(2)} = {\mathcal Q}^{(2)} + \frac{1}{2} \bigl({\mathcal Q}^{(1)}\bigr)^2, \label{a}\\
T^{(3)} = {\mathcal Q}^{(3)} + \frac{1}{2} \bigl( {\mathcal Q}^{(1)} {\mathcal Q}^{(2)} +
 {\mathcal Q}^{(2)} {\mathcal Q}^{(1)}\bigr) +\frac{1}{3!} \bigl({\mathcal Q}^{(1)}\bigr)^3, \qquad \dots\,. \label{c}
\end{gather}

Every element \smash{${\mathcal Q}^{(m)}$} of the series expansion can be obtained
by means of a generic recursive formula (see, for instance,~\cite{Bur, Sal} and~\cite{review} and references therein) as follows.
Consider
\[
\Pi^{(n)}_k = \sum_{j_1+ j_2 + \dots+ j_k =n} T^{(j_1)} \cdots T^{(j_k)},
\]
which satisfies the recursion formula
\begin{equation}
\Pi^{(n)}_k = \sum_{m=1}^{n-k+1} \Pi^{(m)}_1 \Pi^{(n-m)}_{k-1} \qquad \mbox{and} \qquad
\Pi^{(n)}_1 = T^{(n)}, \quad \Pi^{(n)}_n = (T^{(1)})^n. \label{R2}
\end{equation}
Then ${\mathcal Q}^{(m)}$ are given as
\begin{gather}
{\mathcal Q}^{(m)} = T^{(m)} - \sum_{k=2}^m (-1)^k \frac{\Pi^{(m)}_k}{k}, \qquad m>1. \label{R1}
\end{gather}
See, for example, expressions~(\ref{a00}) and~(\ref{ac}) for $m=1, 2$ and $3$.

Hence, we obtain, after recalling~(\ref{R1}), (\ref{R2}) (or specifically~(\ref{a00}),~(\ref{ac})) and~(\ref{12b}), the explicit expressions for the first three terms of the series ({\it Magnus expansion}):
\begin{gather}
{\mathcal Q}^{(1)}(x) = \int_{x_0}^x {\rm d}x_1 A(x_1), \label{qb2a} \\
{\mathcal Q}^{(2)}(x) = \frac{1}{2} \int_{x_0}^x{\rm d}x_2 \int_{x_0}^{x_2} {\rm d}x_1 [A(x_2), A(x_1) ], \nonumber \\
 {\mathcal Q}^{(3)}(x) = \frac{1}{6} \int_{x_0}^x{\rm d}x_3 \int_{x_0}^{x_3} {\rm d}x_2 \int_{x_0}^{x_2} {\rm d}x_1 ( [ A(x_3), [ A(x_2), A(x_1) ] +
 [ [ A(x_3), A(x_2) ], A(x_1) ] ), \nonumber
\end{gather}
where $ [A, B ] := AB - BA$ is the familiar Lie commutator.
\begin{rem}
Magnus obtained in~\cite{Magnus} the general expression for ${\mathcal Q}(x)$ recursively as an infinite series involving Bernoulli's numbers (see also~\cite{review} and references therein). We define
\begin{equation*}
\mbox{ad}_AB = [ A, B ], \qquad \mbox{ad}^n_A = \bigl[A, \mbox{ad}^{n-1}_AB \bigr],
\qquad \mbox{ad}^0_AB = B,
\end{equation*}
and recall the Bernoulli numbers $B_n$ defined as \smash{$\sum_{n=0}^\infty \frac{B_n}{n!}z^n= \frac{z}{{\rm e}^z -1}$}.
Then ${\mathcal Q}(x)$ can be expressed in a compact form as
\begin{equation}
{\mathcal Q}(x) = \int_{x_0} ^x {\rm d}s \sum_{n=0}^{\infty} \frac{B_n}{n!}\mbox{ad}^n_{{\mathcal Q}(x)}A(s). \label{iter}
\end{equation}
Expressions~(\ref{qb2a}) can be obtained from~(\ref{iter}) by iteration. For a detailed discussion on Magnus expansion, convergence issues, expansion generators and applications the interested reader is referred, for instance, to~\cite{review, EbraMan, Fer} and references therein.
\end{rem}

After the brief review on Magnus expansion, we recall in the subsequent section some of the fundamental notions necessary of our analysis in Sections~\ref{section3} and~\ref{section4}. More precisely, we describe below what is achieved in each section.
\begin{itemize}\itemsep=0pt
\item In Section~\ref{section2.1}, we recall the definitions of Rota--Baxter, pre-Lie and tridendriform algebras and we then discuss the connections among these algebras. To illustrate these relations, we use two simple examples, which appear in literature and will be exploited anyway in this study. In Section~\ref{section2.2}, in order to further motivate the study of deep interconnections among seemingly distinct algebraic structures, we recall the passage from pre-Lie algebras to braces (radical rings)~\cite{Rump, Smok1, Smok2}.

 \item In Section~\ref{section3}, the rigorous derivation of the discrete analogue of Magnus expansion and discrete pre-Lie Magnus expansion are exhibited, after having first derived the discrete version of Dyson series taking into account contributions of extra higher-order terms in the formal series expansion in powers of $\alpha$ that do not appear neither in the continuous case nor in the linearized discrete version~\cite{EbraMan}, when dealing with iterated integrals or sums respectively. These derivations are realized by means of the discrete analogue of the evolution problem~(\ref{evo1}). Linearization of the discrete evolution problem leads naturally to the continuous equation~(\ref{evo1}). Furthermore, we extract the discrete Dyson series, which are expressed in terms of a tridendriform algebra binary operation, whereas the explicit discrete pre-Lie Magnus expansion is derived.
 We emphasize that the novelty in our identification is the presence of extra higher-order terms that never appear in the continuous or linear discrete case. These additional terms are crucial when identifying the $N$-coproducts of the Yangian as tridendriform and pre-Lie algebras. Construction of a brace multiplication from this pre-Lie algebra immediately follows. Explicit expressions of the first few elements of both expansions are provided. Taking the continuum limits of these expressions, we recover the continuous Magnus expansion and the known pre-Lie Magnus formula. Also the linear discrete expansion~\cite{EbraMan} is recovered in Remark~\ref{remark12}. In~Section~\ref{section3.1}, we consider alternative discrete Dyson and Magnus expansions, which again can be expressed using tridendriform and pre-Lie algebras. The various discrete expansions are associated to distinct quantum algebras as will be transparent in Section~\ref{section4}. In~Section~\ref{section3.2}, we present two basic examples/applications related to the backward and forward Dyson and Magnus expansions (see related equations~(\ref{discrete}) and~(\ref{discrete2}) later in the manuscript). The first example is related to gauge transformations of matrix valued fields, whereas the second one describes the discrete and continuous evolution problem when folding is also considered (see equation~(\ref{open})).

\item In Section~\ref{section4}, we investigate {interesting links} between quantum algebras, specifically the Yangian and tridendriform and pre-Lie algebras. We first recall the derivation of the Yangian via the Faddeev--Reshetikhin--Takhtajan (FRT) construction. We extract an alternative set of generators using the Lie exponential of the solutions of the FRT relation and derive the defining algebraic relations of the alternative set of generators. We then move on to study tensor realizations of the Yangian for both sets of generators and express coproducts of the Yangian generators using binary operations of pre-Lie and tridendriform algebras. The classical Yangian is briefly discussed after we introduce the notions of the classical $r$-matrix and Sklyanin's bracket~\cite{FT}. We conclude that the non-local charges of the classical Yangian are naturally expressed in terms of (tri)dendriform algebra operations.\looseness=-1
\end{itemize}

We note that part of this work was presented in the ``XL Workshop on Geometric Methods in Physics 2023'', Bia{\l}owie\.{z}a, Poland, and also appears in~\cite{DoikouBal}.

\section{Preliminaries: a short review}
\subsection{Rota--Baxter, pre-Lie and tridendriform algebras}\label{section2.1}

 This section serves as a brief review on various algebraic notions and interlinks that are required for our analysis here, providing also the needed basic background to readers who are not necessarily experts in these research areas. Indeed, we briefly recall the definitions of chronological, pre-Lie algebras~\cite{Chrono, PreLie} and the pre-Lie Magnus expansion, as well as the notion of Rota--Baxter algebras.
The various links between Rota--Baxter algebras, tridendriform and pre-Lie algebras as well as some useful, for our purposes, known propositions and examples are recalled. What follows has been taken from various sources within a rather large bibliography on the mentioned subjects as well as on connections with rooted trees, quasi-shuffle algebras and non-commutative stochastic calculus, enveloping pre-Lie algebras and Magnus expansion, etc. (we refer the interested reader for a more detailed account on these and related matters to~\cite{Bai, Ebra2, new5,Fer, EbraMan,new3, new4, Manchon} and~\cite{tree1, tree2} among others).

Before we comment on the various intriguing interconnections, we first introduce the definitions of Rota--Baxter and pre-Lie algebras~\cite{Rota, RBLoday, Fer, Rota2, Rota3}.
\begin{defn}
A Rota--Baxter algebra is a unital, associative $k$-algebra $\mathcal{A}$ equipped {with a~binary operation $m\colon \mathcal{A} \times \mathcal{A} \to \mathcal{A}$, $(a,b) \mapsto a b$,} and a linear map $R\colon \mathcal{A} \to \mathcal{A}$, such that for all~${a,b \in \mathcal{A}}$
\begin{equation}
R(a) R(b) = R\bigl(R(a) b + aR(b) + \theta ab\bigr), \label{RB}
\end{equation}
where $\theta \in k$ is a fixed parameter.
\end{defn}

The map $R$ is called a Rota--Baxter operator of weight $\theta$.

\begin{defn}
A left pre-Lie algebra is a $k$-vector space $A$ with a binary operation
$
\triangleright\colon (a,b) \mapsto a \triangleright b$ such that it satisfies the left pre-Lie identity $\forall a,b,c \in A$,
\begin{equation}
(a \triangleright b) \triangleright c- a \triangleright(b \triangleright c) = (b \triangleright a) \triangleright c- b \triangleright(a \triangleright c). \label{prelie}
\end{equation}
 Analogously, a right pre-Lie algebra with a binary operation $\triangleleft$ can be defined with a pre-Lie identity
\begin{equation*}
(a \triangleleft b) \triangleleft c- a \triangleleft (b \triangleleft c) = (a \triangleleft c) \triangleleft b- a \triangleleft(c \triangleleft b). 
\end{equation*}
\end{defn}

In the following proposition, it is shown that Rota--Baxter operators of weight $\theta$ can be used to construct a pre-Lie algebra (see~\cite{Fer, EbraMan}).
\begin{pro}\label{pro1-1}
Let ${\mathcal A}$ be a Rota--Baxter
algebra of weight $\theta$, whose Rota--Baxter operator is denoted by $R$. Let also $\triangleright$ be a
binary operation on ${\mathcal A}$ defined for all $a,b\in \mathcal{A}$ as
\begin{equation*}
 a\triangleright b : = [R(a), b ] + \theta a b. 
\end{equation*}
Then $(\mathcal{A}, \triangleright)$ is a left pre-Lie algebra.
\end{pro}
\begin{proof}
It suffices to show the left pre-Lie identity~(\ref{prelie}), indeed we compute
\begin{align}
(x\triangleright y) \triangleright z={}& R (R(x)y - y R(x) + \theta xy )z -z R (R(x)y - yR(x) +\theta xy ) \cr
&{}{+}\, \theta (R(x)yz - yR(x)z + \theta xyz ). \label{rel1}
\end{align}
Similarly, using~(\ref{RB})
\begin{align}
&x\triangleright (y \triangleright z) =R (R(x)y+x R(y) + \theta xy )z +z R (R(y)x +yR(x) +\theta yx ) \label{rel1b}\\
&\qquad{}+ \theta (R(x)yz - yzR(x) + xR(y)z -xzR(y) +\theta xyz ) -R(x)zR(y)- R(y)z R(x). \nonumber
\end{align}
And from the two last expressions, we conclude that
\begin{align*}
(x\triangleright y) \triangleright z - x\triangleright (y \triangleright z) ={}&{-}\,
R (R(x) y + xR(y) )z - zR (R(x) y + xR(y) )\cr
&{-}\,\theta z (R(xy) + R(yx) ) -
\theta (y R(x) + x R(y) ) \cr
&{+}\, \theta (y z R(x) + xzR(y) ) + R(x)zR(y) + R(y)z R(x). 
\end{align*}
The latter expression is symmetric in $x$ and $y$ and leads to the left pre-Lie algebra identity.
\end{proof}

We present two simple examples of Rota--Baxter operators
(other examples can be found, for instance, in
\cite{EbraMan, new3, new4}), which will be
used in our analysis.
\begin{exam} \label{ex0}
A simple example of Rota--Baxter operator is given by the ordinary Riemann integral, which is a weight
zero Rota--Baxter map. Indeed, let
$S(f)_x := \int_{0}^x f(\zeta) {\rm d}\zeta$, then
\begin{align*}
S(f)_x S(g)_x &= \int_{0}^x f(\zeta) \,{\rm d}\zeta \int_0^x g(\xi) \,{\rm d}\xi
= \int _0^x {\rm d}\zeta \int_0^{\zeta} {\rm d}\xi f(\zeta) g(\xi) + \int _0^x {\rm d}\zeta \int_0^{\zeta}{\rm d}\xi f(\xi) g(\zeta) \cr
&= S\bigl(S(f)_{\zeta}g_{\zeta} + f_{\zeta} S(g)_{\zeta} \bigr)_x,
\end{align*}
i.e., $S$ is a Rota--Baxter operator of weight zero.

Let $A$, $B$ be linear operators that depend on a continuous parameter $x$, we define
\begin{equation}
(A \triangleright B)(x) : = \biggl[ \int_0^x A(s) {\rm d}s, B(x) \biggr], \label{bina1}
\end{equation}
which provides a non-commutative binary operation, e.g., $A$ and $B$
can be matrix valued functions of $x$.
It turns out according to Proposition~\ref{pro1-1} that the binary operation
defined in~(\ref{bina1}) satisfies the pre-Lie identity~(\ref{prelie}).

The elements of the Magnus expansion can be re-expressed in terms of a pre-Lie algebra~-- {\it pre-Lie Magnus expansion}~-- (see, e.g.,~\cite{Fer, EbraMan, new3,new4}).
Using expressions
(\ref{qb2a}) and~(\ref{bina1}), we show for the few first terms:
\begin{gather}
{\mathcal Q}^{(1)}(x) = \int_{x_0}^x {\rm d}x_1 A(x_1), \cr
{\mathcal Q}^{(2)}(x) = - \frac{1}{2} \int_{x_0}^x {\rm d}x_2 (A\triangleright A )(x_2),\cr
 {\mathcal Q}^{(3)}(x) = \int_{x_0}^x{\rm d}x_3 \biggl(\frac{1}{4}((A \triangleright A )
\triangleright A)(x_3) + \frac{1}{12} (A \triangleright (A \triangleright A))
(x_3) \biggr). \label{qb2b}
\end{gather}
\end{exam}

\begin{rem} \label{class} There are two classes of linear operators that we are most interested in, especially in Sections~\ref{section3} and~\ref{section4}. We first recall the ${\mathcal N} \times {\mathcal N}$ matrices $e_{a,b}$ with elements $(e_{a,b})_{c,d} = \delta_{a,c}\delta_{b,d}$, $a,b,c,d \in \{1,2, \dots, {\mathcal N}\}$.
\begin{enumerate}\itemsep=0pt
\item \smash{$A_n = \sum_{a,b=1}^{\mathcal N} e_{a,b}(A_{a,b})_n \in \EEnd\bigl({\mathbb C}^{\mathcal N}\bigr)$}, where $(A_{a,b})_n \in {\mathbb C}$ depend on the discrete parameter~${n\in {\mathbb Z}^+}$.
\item Let \smash{$A = \sum_{a,b=1}^{\mathcal N} e_{a,b} \otimes A_{a,b} \in \EEnd\bigl({\mathbb C}^{\mathcal N}\bigr) \otimes {\mathfrak A}$}, where ${\mathfrak A}$ is in general some unital, associative, ${\mathbb C}$-algebra generated by indeterminates $A_{a,b}$, $a,b \in \{1,2,\dots, {\mathcal N}\}$. We then define (in the so-called ``index notation''), \smash{$A_n = \sum_{a,b=1}^{\mathcal N} e_{a,b} \otimes (A_{a,b})_n \in \EEnd\bigl({\mathbb C}^{\mathcal N}\bigr) \otimes {\mathfrak A}^{\otimes N}$}, where ${(A_{a,b})_n \in {\mathfrak A}^{\otimes N}}$ is defined as
\begin{equation}
(A_{a,b})_n := 1_{\mathfrak A} \otimes \dots \otimes 1_{\mathfrak A} \otimes A_{a,b} \otimes 1_{\mathfrak A}\otimes \dots \otimes 1_{\mathfrak A}, \label{not0}
\end{equation}
 $A_{a,b}$ is located on the $n$-th position of the $N$-tensor product of the algebra in~(\ref{not0}) (there are $N$ terms in the tensor product~(\ref{not0}), $1\leq n\leq N$. We note that ${\mathfrak A} \otimes {\mathfrak A}$ is equipped with its usual tensor product algebra
structure: $(a \otimes b)(c \otimes d) = ac \otimes bd$, for all $a, b, c, d \in {\mathfrak A}$. In the special case, where ${\mathfrak A}$ is a commutative algebra, (2)~is equivalent to~(1).
\end{enumerate}
\end{rem}

\begin{exam} \label{ex2}
 Let $f_m$, $g_m $ be linear operators of the type described in Remark~\ref{class}. Define also the sum, \smash{$\Sigma(f)_n : =
\sum_{m=1}^{n-1} f_m$} (see also, e.g.,~\cite{EbraMan}). Then{\samepage
\begin{equation}
\Sigma(f)_n\Sigma(g)_n = \Sigma(\Sigma(f)_mg_m + f_m \Sigma(g)_m +f_mg_m )_n, \label{rota}
\end{equation}
i.e., $\Sigma$ is a Rota--Baxter operator of weight one. This is the discrete analogue of Example~\ref{ex0}.}

Indeed, we rewrite equation~(\ref{rota}) with the suitable discrete variables:
\begin{align*}
\Sigma(f)_n \Sigma(g)_n:={}& \sum_{m=1}^{n-1} f_m \sum_{k=1}^{n-1} g_k = \sum_{m> k}f_m g_k +
\sum_{m< k} f_k g_m +\sum_{m=1}^{n-1} f_ng_n \\
={}&\sum_{m=1}^{n-1} (f_m \Sigma(g)_m + \Sigma(f)_m g_m +f_m g_m ) \\
={}& \Sigma(f_m \Sigma(g)_m+ \Sigma(f)_m g _m +f_mg_m )_n,
\end{align*}
and this concludes the proof.

Note that summation up to $n$ (instead of $n-1$) leads to a Rota--Baxter operator of weight~$-1$.
\end{exam}

\begin{cor} \label{corb} Let $A_m$, $B_m$ be linear operators of the type described in Remark~$\ref{class}$. Define the binary operation
$ \triangleright\colon (A, B) \mapsto A \triangleright B$ such that
\begin{equation}
(A\triangleright B)_n := \Biggl[ \sum_{m=1}^{n-1} A_m, B_n \Biggr] +A_n B_n. \label{bina2}
\end{equation}
Then the left pre-Lie identity is satisfied, i.e.,
\begin{equation*}
 ((A\triangleright B)\triangleright C )_n - (A\triangleright (B\triangleright C) )_n =
 ((B\triangleright A)\triangleright C )_n - (B\triangleright (A\triangleright C) )_n. 
\end{equation*}
\end{cor}
\begin{proof} The proof is immediate by means of Proposition~\ref{pro1-1} and Example~\ref{ex2}.
\end{proof}

As mentioned in the introduction, interesting connections between pre-Lie algebras, rooted tree graphs, tridendriform and Rota--Baxter algebras already well known (see, for instance, \mbox{\cite{Chap1, Fer}}), whereas links between tridendriform, Rota--Baxter and (quasi)shuffle algebras have been also studied in~\cite{LodaySh}. We shall briefly recall now the relation between pre-Lie algebras and (tri)-dendriform algebras.
\begin{defn}
A {\it tridendriform algebra} ${\mathcal D}$ is a $k$-vector space equipped with three binary operations $\prec$, $\succ$, $\cdot$ and the following axioms~\cite{Dendri, Loday}:
\begin{enumerate}\itemsep=0pt
\item[(1)] $(a \prec b) \prec c = a \prec (b\prec c + b\succ c + b\cdot c)$,
\item[(2)] $(a \succ b) \prec c = a \succ (b\prec c)$,
\item[(3)] $a \succ(b \succ c) = (a \prec b + a\succ b + a\cdot b)\succ c$,
\item[(4)] $a\cdot (b \cdot c) = (a\cdot b)\cdot c$,
\item[(5)] $(a\succ b) \cdot c = a\succ( b\cdot c)$,
\item[(6)] $(a\prec b) \cdot c = a\cdot(b\succ c)$,
\item[(7)] $(a\cdot b) \prec c = a\cdot (b\prec c)$.
\end{enumerate}
\end{defn}

 A {\it dendriform algebra} is defined by setting the product $\cdot$ to zero in the above axioms,
consequently the rules of a dendriform algebra are given in terms of axioms (1)--(3) without the~$\cdot$ term.

\begin{rem} Let ${\mathcal D}$ be a tridendriform algebra.
We first note that the binary operation, $*\colon {\mathcal D} \times {\mathcal D} \to {\mathcal D}$ such that
\[
x*y := x\prec y + x\succ y + x\cdot y
\]
for all $x,y \in {\mathcal D}$ is associative~\cite{Loday}. Moreover, $({\mathcal D}, \triangleright)$, $({\mathcal D}, \triangleleft)$ are left, right pre-Lie algebras respectively with
\[
x\triangleright y = x\succ y - y\prec x + x\cdot y, \qquad x\triangleleft y = x\prec y-y\succ x +x\cdot y
\]
for all $x,y \in {\mathcal D}$. In the case of a dendriform algebra, the term $x\cdot y$ is not present in all expression above, i.e., for all $x,y \in {\mathcal D}$
\[
x* y := x\prec y + x\succ y
\]
and
\[
x\triangleright y := x\succ y - y\prec x, \qquad x\triangleleft y := x\prec y-y\succ x .
\]

Let $\mathcal{A}$ be a Rota--Baxter algebra, then $(\mathcal{A},\succ,\prec, \cdot)$ is a tridendriform algebra with~\cite{RBLoday}
\begin{equation}
x\succ y := R(x) y, \qquad x\prec y := x R(y), \qquad a\cdot b = \theta ab \label{action2}
\end{equation}
for all $x,y\in \mathcal{A}$, which also immediately leads to the findings of Proposition~\ref{pro1-1}.
\end{rem}

\subsection{From pre-Lie algebras to braces}\label{section2.2}
We report in this section the recent findings on the relations between pre-Lie algebras and braces~\cite{Rump, Smok1, Smok2}, and in particular the passage from pre-Lie algebras to finite braces. Before we discuss this passage we recall the definition of a brace~\cite{Jesp, Rump1} and we note that braces were essentially introduced in order to derive set-theoretic solutions of the Yang--Baxter equation~\cite{Rump1}.
The already known relationships between braces, the Yang--Baxter equation and quantum integrability~\cite{DoiSmo2, DoiSmo1}, as well as the passage described below are expected to lead to even deeper associations and the study of possibly novel algebraic structures.

\begin{defn}
A {\it left brace} is a set $B$ together with two group operations $+,\circ\colon B\times B\to B$,
the first is called addition and the second is called multiplication such that $\forall a,b,c\in B$
\begin{equation*}
a\circ (b+c)=a\circ b-a+a\circ c.
\end{equation*}
\end{defn}

The additive identity of a left brace $B$ will be denoted by $0$ and the multiplicative identity by~$1$, and in every left brace $0=1$. Also, let $(N,+, \cdot)$ be an {associative, nilpotent} ring. For~${a, b \in N}$, define $a \circ b := a \cdot b + a + b$, then $(N,+, \circ )$ is a brace~\cite{Rump1}.

The group of formal flows constructed from a pre-Lie algebra was introduced in~\cite{Chrono} (see also, e.g.,~\cite{Manchon}).
We summarize below the passage from pre-Lie algebras {to finite braces}~\cite{Rump, Smok1, Smok2} after recalling the definition of the
group of formal flows~\cite{Chrono, Manchon}. We also assume as in~\cite{Smok1, Smok2} that ${\mathrm A}$ is a nilpotent pre-Lie algebra.
\begin{enumerate}\itemsep=0pt
\item Let $a\in {\mathrm A}$, and let $L_{a}\colon {\mathrm A}\rightarrow{\mathrm A}$ denote the left multiplication by $a$, so
$L_{a}(b):=a\triangleright b$, also $L_{c}(L_{b}(a))=c\triangleright (b\triangleright a)$,
 and
 \[
 {\rm e}^{L_{a}}(b) =b+a\triangleright b+{\frac 1{2!}}a\triangleright (a\triangleright b)+{\frac 1{3!}}a\triangleright (a\triangleright (a\triangleright b))+\cdots .
 \]

\item We formally consider the element ${\bf 1}$, such that ${\bf 1}\triangleright a=a\triangleright {\bf 1}=a$ in the pre-Lie algebra (as in~\cite{Manchon}) and
define
\[
W(a):={\rm e}^{L_{a}}({\bf 1})-{\bf 1}=a+{\frac 1{2!}}a\triangleright a+{\frac 1{3!}}a\triangleright (a\triangleright a)+ \cdots .
\]
 $W(a)\colon {\mathrm A}\rightarrow {\mathrm A}$ is a bijective function, provided that ${\mathrm A}$ is a nilpotent pre-Lie algebra.

\item Let $\Omega (a)\colon {\mathrm A}\rightarrow {\mathrm A}$ be the inverse function to the function $W(a)$, i.e.,
\[
\Omega (W(a))=W(\Omega (a))=a.
\]
Following~\cite{Manchon}, the first terms of $\Omega $ are
\[
\Omega(a)= a-{\frac 12}a\triangleright a +{\frac 14} (a\triangleright a)\triangleright a +{\frac {1}{12}}a\triangleright(a\triangleright a) +\cdots .
\]

\item We define the multiplication, \[a\circ b:=a+{\rm e}^{L_{\Omega (a)}}(b).\]
The addition is the same as in the pre-Lie algebra ${\mathrm A}$.
It was shown in~\cite{Chrono} that $({\mathrm A}, \circ )$ is a~group. It is then straightforward to show that $({\mathrm A}, \circ , + )$ is a left brace, indeed
 \begin{align*}
 a\circ (b+c)+a&=a+{\rm e}^{L_{\Omega (a)}}(b+c)+a=\bigl(a+{\rm e}^{L_{\Omega (a)}}(b)\bigr)+\bigl(a+{\rm e}^{L_{\Omega (a)}}(c)\bigr)\\
 &=a\circ b+a\circ c.
 \end{align*}
\end{enumerate}

The above formula can also be written using the Baker--Campbell--Hausdorff (BCH) formula, (see~\cite{Chrono, Manchon}). We first recall that the Lie algebra $L({\mathrm A})$ is obtained from a pre-Lie algebra $A$ by defining $[a,b]=a\triangleright b-b\triangleright a$ (with the same addition as ${\mathrm A}$).
By means of the BCH formula, \smash{${\rm e}^{L_a}\bigl({\rm e}^{L_b}({\bf 1})\bigr)= {\rm e}^{L_{C(a,b)}}({\bf 1})$}, the element $C(a, b)$ can be represented in the form of a series as
\[
C(a,b)=a+b+{\frac 12}[a,b]+{\frac 1{12}}([a,[a,b]]+[b,[b,a]])+\cdots.
\]

\begin{lem}The following formula for the multiplication $\circ $ defined above holds $($see, for example,~\mbox{\rm \cite{Chrono, Manchon})}:
\[
W(a)\circ W(b)= W(C(a, b)),
\]
where $C(a, b)$ is obtained using the BCH series in the Lie algebra $L({\mathrm A})$.
\end{lem}
\begin{proof} The proof is immediate from the definition of the brace multiplication
\begin{align*}
W(a) \circ W(b) & = W(a) + {\rm e}^{L_{\Omega(W(a))}}(W(b)) = {\rm e}^{L_{a}}({\bf 1})-{\bf 1}
+ {\rm e}^{L_{a}}\bigl({\rm e}^{L_{b}}({\bf 1})-{\bf 1}\bigr)\cr
& = {\rm e}^{L_{a}}\bigl({\rm e}^{L_b}({\bf 1})\bigr)-{\bf 1}={\rm e}^{L_{C(a,b)}}({\bf 1})-{\bf 1} = W(C(a, b)). \nonumber
\tag*{\qed}
\end{align*}
\renewcommand{\qed}{}
\end{proof}

Some general comments are also in order here. In general, tridendriform algebras yield left pre-Lie algebras, but they also gives rise to so-called post-Lie algebras~\cite{p[2], p[3],p[5]}. The notion of post-Lie Magnus expansion, thoroughly investigated during the last decade, should also be related to the discrete Magnus expansion. For a
detailed account on post-Lie algebras, the interested reader is referred to~\cite{p[2], p[3], Ebf}. Post-Lie algebras are also related to skew-braces~\cite{p[4]} (alias post-groups~\cite{p[1]}) the same way pre-Lie algebras are related to braces.

\section[Discrete Magnus expansion and discrete pre-Lie Magnus expansion]{Discrete Magnus expansion\\ and discrete pre-Lie Magnus expansion}\label{section3}

In this section, the rigorous derivation of the discrete analogue of Magnus expansion and the discrete pre-Lie Magnus expansion are presented. These derivations are based on the discrete analogue of the evolution problem~(\ref{evo1}). We also extract the discrete Dyson series, which as expected from the continuum case are expressed in terms of a tridendriform {algebra}, whereas the discrete Magnus expansion is expressed in terms of a pre-Lie algebra. We should note that the discrete linearized Dyson and Magnus expansions were first derived in~\cite{EbraMan}. In fact, as discussed in Remark~\ref{remark12} the linearized version of our generic expressions reduces to the discrete expressions appearing in~\cite{EbraMan}. Moreover, the characterization of the summation operation as a Rota--Baxter operator together with other useful identities are provided in~\cite{EbraMan}, and this is utilized here as will be transparent in this section ({see Example~\ref{ex2} and Corollary~\ref{corb}}).

The novelty in our derivation concerns the presence of ``higher-order terms'' that never appear in the continuous case or the discrete linearized version derived in~\cite{EbraMan}, when dealing with iterated integrals or sums. These extra terms are due to the choice of the linear operator in the discrete evolution problem and are essential for expressing the $N$-coproducts of the Yangian generators as tridendriform and pre-Lie algebras as described in Section~\ref{section4}, which is one of our key derivations here.
We note that the detailed proof on the continuum
limit of the solution of the discrete evolution problem and the elimination of the extra terms when taking the continuum limit is based on a ``power counting'' argument and is given in~\cite{AvanDoikouSfetsos} (see also Remark~\ref{remcont}).

In Section~\ref{section1}, we recalled the solution of the linear evolution problem expressed as Dyson and Magnus series and we also recalled the pre-Lie Magnus expansion. We are now focusing on the discrete evolution problem and the derivation of the discrete analogue of the Magnus expansion.

Let
\[
{\mathbb L}(\alpha)= 1 + \sum_{m>0} \alpha^m L^{(m)} \in \EEnd\bigl({\mathbb C}^{\mathcal N}\bigr) \otimes {\mathfrak A},
\]
where $\alpha\in {\mathbb C}$,
\[
L^{(m)} = \sum_{a,b =1}^{{\mathcal N}} e_{a,b} \otimes L_{a,b}^{(m)} \in \EEnd\bigl({\mathbb C}^{\mathcal N}\bigr) \otimes {\mathfrak A}
\]
 and ${\mathfrak A}$ is some unital, associative, ${\mathbb C}$-algebra with generators \smash{$L_{a,b}^{(m)}$}, $a,b \in \{1,2 ,\dots ,{\mathcal N}\}$, $m \in {\mathbb Z}^+$ (recall also Remark~\ref{class})). In order to derive the discrete analogue of Dyson and Magnus series, we consider the discrete evolution problem
\begin{equation}
T_{n+1}(\alpha) = {\mathbb L}_n(\alpha) T_{n}(\alpha).\label{discrete}
\end{equation}
We call the discrete evolution problem~(\ref{discrete}) {\it forward}.
The solution of the difference equation above is found by iteration, and is given by {the so-called forward monodromy matrix (or just monodromy), and is the discrete analogue of~(\ref{eq:CTS_T})} (see also Figure~\ref{fig1})
\begin{equation}
 T_{N+1}(\alpha) = {\mathbb L}_{N}(\alpha) \cdots {\mathbb L}_{1}(\alpha
 ) \in \EEnd\bigl({\mathbb C}^{\mathcal N}\bigr) \otimes {\mathfrak A}^{\otimes N} \label{monoad}
 \end{equation}
(let us choose for simplicity $T_1(\alpha) = 1$ as initial condition; see also Figure~\ref{fig1} below).

\begin{figure}[!ht]\centering\vspace{10mm}

\begin{picture}(1,1)
	\put(-90,0){\line(1,0){180}}
	\multiput(-75,30)(0,-8){8}{\line(0,-1){6}}
	\multiput(-45,30)(0,-8){8}{\line(0,-1){6}}
	\multiput(75,30)(0,-8){8}{\line(0,-1){6}}
	\put(-15,-20){$\cdots\cdots\cdots\cdots$}
	\put(-78,-50){$N$}
	\put(-48,-50){$N-1$}
 \put(0, 10){${\mathbf \longrightarrow}$}
	\put(72,-50){$1$}
	\put(-92,5){$ $}
\end{picture}

\vspace{15mm}

\caption{Graphical representation of the monodromy matrix $T_{N+1}$.}\label{fig1}
\end{figure}

The horizontal line in the diagram above represents the vector space ${\mathbb C}^{\mathcal N}$, whereas each intermittent line represents a copy of the algebra ${\mathfrak A}$; this is in agreement with the index notation introduced in (2) in Remark~\ref{class}. Figure~\ref{fig1} reads from left to right following the arrow (path ordered).

In the frame of classical and quantum integrable systems, the trace of the monodromy matrix generates a hierarchy of conserved quantities called Hamiltonians that describe integrable systems with periodic boundary conditions (see, for instance,~\cite{FT, FadTakRes}, see also Remark~\ref{tt} later in the manuscript).

Let also $T_{N+1}$ be expressed as a formal power series \smash{$T_{N+1}(\alpha) =1 + \sum_{m>0} \alpha^n T^{(m)}(N+1)$}, then by considering the generic form \smash{${\mathbb L}(\alpha) = 1 + \sum_{k>1} \alpha^k L^{(k)}$} we obtain the coefficients of the monodromy matrix expansion via~(\ref{monoad}):
\begin{gather}
T^{(m)}(N+1)= \sum_{\sum_{j=1}^km_j=m}\Biggl(\sum_{n=1}^N L_n^{(m_k)} \sum_{n_{k-1}=1}^{n-1 }L^{(m_{k-1})}_{n_{k-1}} \dots \sum_{n_{1}=1}^{n_2-1 }L^{(m_1)}_{n_1}\Biggr). \label{tt2}
\end{gather}
In the following proposition, we express the elements of the discrete analogue of
Dyson's series~\eqref{tt2} in terms of the tridendriform algebra.

\begin{pro}
{\label{lemma01}} The elements of the discrete analogue of Dyson series
\eqref{tt2} are expressed as
\begin{gather}
T^{(m)}(N+1) = \sum_{\sum_{j=1}^km_j=m}\sum_{n=1}^N \bigl(L^{(m_k)}\prec \bigl(L^{(m_{k-1})} \prec \bigl( \cdots \prec\bigl( L^{(m_2)}\prec L^{(m_1)}\bigr)\bigr)\cdots \bigr) \bigr)_n , \label{dendr2}
\end{gather}
where the tridendriform operation $\prec$ is defined as \smash{$(a\prec b)_n:= a_n \Sigma(b)_n$} and \smash{$\Sigma(b)_n = \sum_{m=1}^{n-1}b_m$}, \smash{$a,b \in \EEnd\bigl({\mathbb C}^{\mathcal N}\bigr) \otimes {\mathfrak A}$} $($see also Remark~$\ref{class})$.
\end{pro}
\begin{proof}
Using the definitions of $\prec$ and $\Sigma$ and the explicit expressions of $T^{(m)}$, $m \in \{1,2, \dots, N\}$, (\ref{tt2}), we write
\begin{equation}
T^{(m)}(N+1) =\sum_{\sum_{j=1}^km_j=m} \sum_{n=1}^N L^{(m_k)}_n \Sigma\bigl(L^{(m_{k-1})}\Sigma\bigl(L^{(m_{k-2})}\Sigma\bigl(\cdots \Sigma\bigl(L^{(m_1)}\bigr)\bigr)\cdots \bigr)\bigr)_n. \label{sigma}
\end{equation}
Via \smash{$(x\prec y)_n=x_n \Sigma(y)_n$}, expression~(\ref{sigma}) leads to~(\ref{dendr2}).
\end{proof}

For instance, expressions $T^{(m)}(N+1)$~(\ref{dendr2}) become: for $m=1$,
\[
T^{(1)}(N+1)= \sum_{n=1}^N L^{(1)}_n,
\]
for~${m=2}$
\begin{align*}
T^{(2)}(N+1) &= \sum_{n=1}^N \Biggl(L^{(1)}_n \sum_{m=1}^{n-1} L^{(1)}_m + L^{(2)}_n \Biggr)= \sum_{n=1}^N \Biggl(L^{(1)}_n \Sigma\bigl(L^{(1)}\bigr)_n + L^{(2)}_n \Biggr)\cr
&=
 \sum_{n=1}^N \bigl( \bigl(L^{(1)} \prec L^{(1)}\bigr)_n + L_n^{(2)}\bigr)\nonumber
\end{align*}
for $m=3$,
\begin{align*}
T^{(3)}(N+1) &= \sum_{n=1}^N\Biggl( L^{(1)}_n \sum_{m=1}^{n-1} L^{(1)}_m \sum_{k=1}^{m-1} L^{(1)}_k + L^{(1)}_n \sum_{m=1}^{n-1} L^{(2)}_m + L^{(2)}_n \sum_{m=1}^{n-1}L^{(1)}_m + L_n^{(3)}\Biggr)\cr
&= \sum_{n=1}^N \bigl(L^{(1)}_n \Sigma\bigl(L^{(1)}\Sigma\bigl(L^{(1)}\bigr)\bigr)_n +L^{(1)}_n \Sigma\bigl(L^{(2)}\bigr)_n +L_n^{(2)}\Sigma\bigl(L^{(1)}\bigr)_n +L_n^{(3)} \bigr)
\cr
&=
 \sum_{n=1}^N \bigl( \bigl(L^{(1)} \prec \bigl(L^{(1)}\prec L^{(1)}\bigr)\bigr)_n +\bigl(L^{(1)}\prec L^{(2)} \bigr)_n +
\bigl(L^{(2)} \prec L^{(1)}\bigr)_n + L_n^{(3)}\bigr). \nonumber
\end{align*}

We come now to the precise derivation of the discrete Magnus expansion. We consider the Lie exponential \smash{$T_{N+1}(\alpha) = {\rm e}^{{\mathcal Q}_{N+1}(\alpha)}$}, \smash{${\mathcal Q}_{N+1}(\alpha) = \sum_{m=1}^{\infty}\alpha^m {\mathcal Q}^{(m)}(N+1)$}, which will lead to the discrete analogue of Magnus expansion; expressions~(\ref{a00}),~(\ref{ac}),~(\ref{R1}) and (\ref{R2}) hold.

\begin{lem}\label{lemma2} Let
\[
T_{N+1}(\alpha) = {\rm e}^{{\mathcal Q}_{N+1}(\alpha)}, \qquad T_{N+1}(\alpha)= 1+ \sum_{m=1}^{\infty}\alpha^m T^{(m)}(N+1),
\]
 where $T^{(m)}({N+1})$ are given in Proposition~$\ref{lemma01}$ and \smash{${\mathcal Q}_{N+1}(\alpha) = \sum_{m=1}^{\infty}\alpha^m {\mathcal Q}^{(m)}(N+1)$}. The quantities
\[
{\mathcal Q}^{(k)}(N+1):= \sum_{n=1}^N \Omega^{(k)}_n,\]
where \smash{$\Omega^{(k)}_n = {\mathcal Q}^{(k)}(n+1) - {\mathcal Q}^{(k)}(n)$}, are expressed explicitly as
\begin{gather}
{\mathcal Q}^{(1)}(N+1) = \sum_{n=1}^{N} L^{(1)}_{n},\nonumber\\
{\mathcal Q}^{(2)}(N+1) = \frac{1}{2}\sum_{n>n_1=1}^{N} \bigl[ L^{(1)}_{n}, L^{(1)}_{n_1} \bigr] -
\frac{1}{2}\sum_{n=1}^{N} \bigl(L^{(1)}_{n}\bigr)^2+\sum_{n=1}^NL_n^{(2)},
\nonumber\\
{\mathcal Q}^{(3)}(N+1) = \sum_{n=1}^N \Biggl(\frac{1}{6} \sum_{n_2>n_1=1}^{n-1}
\bigl(\bigl[L^{(1)}_{n}, \bigl[L^{(1)}_{n_2}, L^{(1)}_{n_1} \bigr] \bigr] +
\bigl[\bigl[L^{(1)}_{n}, L^{(1)}_{n_2} \bigr], L^{(1)}_{n_1} \bigr]\bigr) \nonumber\\
\hphantom{{\mathcal Q}^{(3)}(N+1) = }{}
 + \frac{1}{6} \sum_{n_1=1}^{n-1} \bigl( L^{(1)}_{n_1}\bigl[L^{(1)}_{n_1},
L^{(1)}_{n}\bigr] + \bigl[L^{(1)}_{n_1}, L^{(1)}_{n}\bigr] L^{(1)}_{n}\bigr) \nonumber\\
\hphantom{{\mathcal Q}^{(3)}(N+1) = }{} +
\frac{1}{6} \sum_{n_1=1}^{n-1} \bigl(\bigl[L^{(1)}_{n_1}, \bigl(L^{(1)}_{n}\bigr)^2\bigr] +
\bigl[\bigl(L^{(1)}_{n_1}\bigr)^2, L^{(1)}_{n}\bigr] \bigr)-\frac{1}{2} \sum_{n=1}^{N} \bigl(L_n^{(1)} L_n^{(2)} +
L_n^{(2)} L_n^{(1)} \bigr)\nonumber\\
\hphantom{{\mathcal Q}^{(3)}(N+1) = }{}
+ \frac{1}{3} \bigl(L^{(1)}_{n}\bigr)^3 - \frac{1}{2} \sum_{m=1}^{n-1}
\bigl( \bigl[L^{(1)}_m, L_n^{(2)} \bigr] +\bigl[L^{(2)}_m, L^{(1)}_n \bigr] \bigr) + L_n^{(3)} \Biggr), \qquad \dots\,. \label{QQ}
\end{gather}
\end{lem}
\begin{proof}
{Recursive expressions~(\ref{R1}) and~(\ref{R2}) apparently still hold, but now $T^{(m)}$ are given by~(\ref{tt2}).
Notice that in the discrete case both $T_{N+1}(\alpha)$ and $T^{(m)}(N+1)$ depend on a discrete parameter~$N$,
which replaces the continuum parameter $x$ of the continuous analogue discussed in Section~\ref{section1}. The quantities \smash{$\Omega_n^{(k)}$} can be immediately read from~(\ref{QQ}), and are the discrete analogues of the derivatives \smash{$\dot {\mathcal Q}^{(k)}(x)$}, of Magnus expansion in Section~\ref{section1}.}
\end{proof}

\begin{pro} \label{prom}
The elements of the discrete Magnus expansion~\eqref{QQ} can be re-expressed in terms of the pre-Lie operation~\eqref{bina2} as $(${\it discrete pre-Lie Magnus expansion}$)$
\begin{gather}
 {\mathcal Q}^{(1)}(N+1) = \sum_{n=1}^NL^{(1)}_n, \nonumber\\
 {\mathcal Q}^{(2)}(N+1) =-\frac{1}{2} \sum_{n=1}^N
 \bigl(L^{(1)} \triangleright L^{(1)}\bigr)_n +\sum_{n=1}^NL_n^{(2)},\nonumber\\
 {\mathcal Q}^{(3)}(N+1)= \sum_{n=1}^N\biggl(
\frac{1}{4} \bigl(\bigl(L^{(1)} \triangleright L^{(1)}\bigr) \triangleright L^{(1)}\bigr)_n +\frac{1}{12}
\bigl(L^{(1)} \triangleright \bigl(L^{(1)} \triangleright L^{(1)}\bigr)\bigr)_n \biggr)\nonumber\\
\hphantom{ {\mathcal Q}^{(3)}(N+1)=}{}
-\frac{1}{2} \sum_{n=1}^N\bigl(\bigl(L^{(2)} \triangleright L^{(1)}\bigr)_n+ \bigl(L^{(1)}\triangleright L^{(2)}\bigr)_n\bigr)
+ \sum_{n=1}^N L_n^{(3)},\qquad \dots\,. \label{pre1}
\end{gather}
\end{pro}
\begin{proof}
We recall the definition of the pre-Lie binary operation~(\ref{bina2}), then \big($x \in\bigl\{ L^{(1)}, L^{(2)}, \dots\bigr\}$\big)
\begin{gather}
(x \triangleright x) _n= \Biggl[\sum_{m=1}^{n-1}x_m, x_n \Biggr] +
x_n^2, \label{tria1}
\end{gather}
also, from expressions~(\ref{rel1}) and~(\ref{rel1b}), for $R \to \Sigma$, and by setting $x=y=z$, we obtain
\begin{align}
\frac{1}{4}((x \triangleright x) \triangleright x)_n +
\frac{1}{12} (x \triangleright (x \triangleright x))_n ={}& \frac{1}{6} \sum_{n_2>n_1=1}^{n-1}
 ([x_{n}, [x_{n_2}, x_{n_1} ] ] +
[[x_{n}, x_{n_2} ], x_{n_1} ]) \cr
&{+}\, \frac{1}{6} \sum_{n_1=1}^{n-1} ( x_{n_1}[x_{n_1},\
x_{n}] + [x_{n_1}, x_{n}] x_{n} ) \cr
&{+}\, \frac{1}{6} \sum_{n_1=1}^{n-1} \bigl(\bigl[x_{n_1}, x^2_{n}\bigr] +
\bigl[x^2_{n_1}, x_{n}\bigr] \bigr) +
\frac{1}{3} x^3_{n}. \label{tria2}
\end{align}
Comparing expressions~(\ref{QQ}) with~(\ref{tria1}) and~(\ref{tria2}), we arrive at expressions~(\ref{pre1}).
These expressions provide elegant discrete analogues of the pre-Lie Magnus expansion. Higher-order terms~${\mathbb Q}^{(m)}$,~${m>3}$, can be computed by iteration via~(\ref{R1}) and~(\ref{R2}).
We only compute indicatively the first few terms of the expansion. Using the recursive relations~(\ref{R1}) and~(\ref{R2}),
expressions~(\ref{tt2}) and~(\ref{QQ}) and definition~(\ref{tria1}), we may explicitly compute higher-order terms; however, such
computations become increasingly complicated due to the existence of extra terms coming from the higher-order terms $L^{(m)}$, $m>1$,
in the formal expansion of ${\mathbb L}$~(\ref{discrete}) (see, for instance, the last term in ${\mathcal Q}^{(2)}$ the last two terms in line three ${\mathcal Q}^{(3)}$ and the last three terns in the last line of
${\mathcal Q}^{(3)}$,~(\ref{QQ})).
\end{proof}
\begin{rem}{\label{remark12}} It is worth focusing on the simple linear case, where
\[
{\mathbb L}(\alpha) = 1 + \alpha {\mathbb P} \in \EEnd\bigl({\mathbb C}^{\mathcal N}\bigr) \otimes {\mathfrak A}.
\]
The linear case is clearly a special example of the general scenario studied above, where we considered the formal series expansion, \smash{${\mathbb L}(\alpha) =1 + \sum_{m=1}^{\infty}\alpha^m L^{(m)} \in \EEnd\bigl({\mathbb C}^{\mathcal N}\bigr) \otimes {\mathfrak A};$} indeed in the linear case, $L^{(1)} = {\mathbb P}$ and $L^{(m)} =0$, $\forall m>1$. In this case the generic expressions of Propositions~\ref{lemma01} and~\ref{prom} reduce to the linearized discrete expressions obtained in~\cite{EbraMan}.
Indeed, we obtain the following explicit expressions, which provide a simpler discrete analogue of Dyson's series
\begin{gather}
T_{N+1}(\alpha) = 1 + \sum_{m=1}^N\alpha^{m} T^{(m)}(N+1), \qquad T^{(m)}(N+1) =
\sum_{n_m>\dots>n_1=1}^{N} {\mathbb P}_{n_m} \cdots {\mathbb P}_{n_1}. \label{asyt}
\end{gather}
According to Proposition~\ref{lemma01}, the elements of the discrete Dyson series~(\ref{asyt})
are expressed in terms of the tridendriform operation $\prec $ as
\begin{gather}
T^{(m)}(N+1) = \sum_{n =1}^N ({\mathbb P}\prec ({\mathbb P}\prec
({\mathbb P} \prec (\dots \prec({\mathbb P} \prec {\mathbb P})) \cdots )) )_n, \label{dendr1}
\end{gather}
with $m$ ${\mathbb P}$-terms. We also recall that \smash{$T := {\rm e}^{{\mathcal Q}}$}, then via~(\ref{a00}) and~(\ref{ac})
we immediately obtain the formulas for ${\mathcal Q}^{(m)}(N+1)$, given the findings of the generic case,
and the pre-Lie discrete Magnus expansion in this case takes the simple form resembling structurally the continuum case,
\begin{gather}
 {\mathcal Q}^{(1)}(N+1) = \sum_{n=1}^N{\mathbb P}_n, \\
 {\mathcal Q}^{(2)}(N+1) =-\frac{1}{2} \sum_{n=1}^N ( {\mathbb P} \triangleright {\mathbb P})_n, \nonumber\\
 {\mathcal Q}^{(3)}(N+1)= \sum_{n=1}^N\biggl(
\frac{1}{4} (({\mathbb P} \triangleright {\mathbb P}) \triangleright {\mathbb P} )_n + \frac{1}{12} ({\mathbb P} \triangleright ({\mathbb P} \triangleright {\mathbb P}) )_n \biggr).\label{prelie2b}
\end{gather}
Higher-order terms are obtained via~(\ref{R1}) and~(\ref{R2}). Expressions~(\ref{prelie2b}) are much more concise compared to the general ones~(\ref{pre1}), and apparently similar to the corresponding continuous formulas,
given that {terms that contain higher-order elements $L^{(m)}$, $m>1$, in the general case} are missing in the linear scenario.
\end{rem}

We should note that we have not been able to produce general expressions neither of the type~(\ref{iter}) nor in terms of pre-Lie algebras, due to the existence of the higher-order terms~$L^{(m)}$, ${m>1}$ in the expansion of \smash{${\mathbb L}(a) = 1 + \sum_{m\geq 1} \alpha^m L^{(m)}$}. However, in the linear case, where ${{\mathbb L}(a) = 1+ \alpha {\mathbb P}}$ (discussed in Remark~\ref{remark12}) such expressions are available in~\cite{Fer, EbraMan}. This is a~very interesting open problem for the general case, which we hope to address in a future work.

\begin{rem}[the continuum limit]\label{remcont}
Recall the discrete
evolution problem~(\ref{discrete}), and consider ${\mathbb L}(\alpha) \in \EEnd\bigl({\mathbb C}^{\mathcal N}\bigr)$,
then rescale $\alpha \to \delta \alpha$ ($\delta \ll 1 $) and consider the general form
${\mathbb L}(\lambda) = 1 + \smash{\sum_{m>0} \alpha^m \delta^n L^{(m)}}$. Expression~(\ref{discrete})
can be rewritten as (keeping only linear terms in~$\delta$)
\begin{equation*}
T_{n+1}(\alpha) = \bigl(1+ \alpha \delta L^{(1)}_{n+1} \bigr) T_n(\alpha).
\end{equation*}
By considering the following ``dictionary'' as $\delta \to 0$:
$L^{(1)}_{n+1} \to A(x)$, and \smash{$\frac{\Psi_{n+1} - \Psi_n}{\delta} \to \partial_{\xi} \Psi(\xi)$},
we arrive at the continuum limit of~(\ref{discrete}), which is the linear evolution problem
\begin{equation*}
\partial_{\xi} T(\xi, \alpha) = \alpha A(\xi) T(\xi, \alpha).
\end{equation*}
 Detailed proof on the continuum
limit of the discrete monodromy~(\ref{monoad}) (which gives the
continuum monodromy~(\ref{eq:CTS_T})), based on a ``power counting rule'' is given in~\cite{AvanDoikouSfetsos}.
The counting rule relies on the fact, \smash{$\delta \sum_{n=1}^N f_n \to \int_{0}^xf(\xi){\rm d}\xi$}, and terms of the form
\[
\delta^{\sum_{j=1}^m n_j}\sum_{k_1, k_2,\dots ,k_l} L_{k_1}^{(n_1)}\cdots L_{k_m}^{(n_{m})} \to 0
\]
in the continuum limit for $\sum_{j=1}^m n_j > m$.
\end{rem}

In the continuum limit, the monodromy matrix $T$ is expressed as a Dyson series with terms that are written, via~(\ref{asyt}) and~(\ref{dendr1}) in terms of a dendriform binary operation $ (A\prec B)(x) := \smash{A(x) \int_0^x B(\xi) {\rm d}\xi}$, as
\begin{eqnarray*}
T^{(m)}(x) =\int_{0}^x{\rm d}\zeta (A\prec (A\prec
(A\prec (\cdots \prec(A \prec A)) \cdots )))(\zeta).
\end{eqnarray*}

\subsection{An alternative discrete expansion}\label{section3.1}
 {Let us now consider a slightly different scenario, where the ${\mathbb L}$-operator in~(\ref{discrete}) is of the form ${\mathbb L}(\alpha) = {\mathrm M} + \alpha {\mathrm L}$, ${\mathrm M} \neq 1$ is invertible and
\[
{\mathrm M} = \sum_{a,b =1}^{\mathcal N} e_{a,b} \otimes M_{a,b} \in \EEnd\bigl({\mathbb C}^{\mathcal N}\bigr) \otimes {\mathfrak A},\qquad
{\mathrm L} = \sum_{a,b =1}^{\mathcal N} e_{a,b} \otimes L_{a,b} \in \EEnd\bigl({\mathbb C}^{\mathcal N}\bigr) \otimes {\mathfrak A},
\]
 where in this section ${\mathfrak A}$ is a unital, associative ${\mathbb C}$-algebra generated by $M_{a,b}$ and $L_{a,b}$ with ${a,b \in \{1,2,\dots, {\mathcal N}\}}$.}

Before we state the main findings in the next proposition, it is useful to introduce some notation. We set (recall that the index notation has been introduced in Remark~\ref{class})
\begin{enumerate}\itemsep=0pt
\item[(1)] ${\mathbb M}_{N+1, n-1}: = {\mathrm M}_N \cdots {\mathrm M}_{n} {\mathrm M}_{n-1}$,
\item[(2)] $\hat {\mathrm M} := {\mathrm M}^{-1}$ and $\hat {\mathbb M}_{n-1, N+1} =\hat {\mathrm M}_{n-1} \hat {\mathrm M}_{n} \cdots \hat {\mathrm M}_N$,
\item[(3)] ${\mathbb M}_{N+1, 1} :={\mathbb M}_{N+1}$ and $\hat {\mathbb M}_{1,N+1} := \hat {\mathbb M}_{N+1}$.
\end{enumerate}

\begin{pro}
Let ${\mathbb L}(\alpha)\in \EEnd\bigl({\mathbb C}^{\mathcal N}\bigr) \otimes {\mathfrak A}$ in~\eqref{discrete} be of the form ${\mathbb L}(\alpha) = {\mathrm M} + \alpha {\mathrm L}$, $\alpha \in {\mathbb C}$, where ${\mathrm M}\neq 1$, ${\mathrm L}\in \EEnd\bigl({\mathbb C}^{\mathcal N}\bigr) \otimes {\mathfrak A}$. Then the monodromy matrix $T_{N+1}(\alpha)$ can be expressed as
\begin{equation*}
T_{N+1}(\alpha) ={\rm e}^{{\mathcal Q}_{N+1}(\alpha)}{\mathbb M}_{N+1},
\end{equation*}
where \smash{$T_{N+1}(\alpha) = \sum_{n=0}^N \alpha^n T^{(n)}(N+1)$}, \smash{${\mathcal Q}_{N+1}(\alpha) = \sum_{m=1}^{N+1} \alpha^m {\mathcal Q}^{(m)}(N+1)$} is the discrete Magnus expansion as in Remark~$\ref{remark12}$ and the elements ${\mathcal Q}^{(m)}(N+1)$ are given in~\eqref{prelie2b} as the pre-Lie Magnus expansion.
\end{pro}
\begin{proof}

We start with the standard ordered expansion of the monodromy matrix, given the choice ${\mathbb L}(\alpha)= {\mathrm M} + \alpha {\mathrm L}$, we obtain
\begin{align*}
T_{N+1}(\alpha)={}& {\mathbb M}_{N+1} + \alpha\sum_{n=1}^N {\mathrm M}_{N} \cdots {\mathrm M}_{n+1} {\mathrm L}_n {\mathrm M}_{n-1} \cdots {\mathrm M}_1 \nonumber\\
&{}{+}\, \alpha^2 \sum_{n>m=1}^N {\mathrm M}_N \cdots {\mathrm M}_{n+1}{\mathrm L}_n {\mathrm M}_{n-1} \cdots {\mathrm M}_{m+1} {\mathrm L}_m {\mathrm M}_{m-1} \cdots {\mathrm M}_1 + \cdots
\end{align*}
and after defining
\begin{equation}
{\bf P}_n := {\mathbb M}_{N+1,n-1}{\mathrm L}_n \hat {\mathrm M}_n \hat {\mathbb M}_{n-1, N+1}, \label{alt0}
\end{equation}
we arrive at
\begin{gather}
T_{N+1}(\alpha) = \Biggl(1+ \sum_{m=1}^N\alpha^m \sum_{n_m>\dots >n_1=1}^N{\bf P}_{n_m} {\bf P}_{n_{m-1}} \cdots {\bf P}_{n_1} \Biggr){\mathbb M}_{N+1}, \label{alt}
\end{gather}
where ${\mathbb M}_n$ is derived in (1)--(3) in the beginning of the section.

The bracket in the expression above is just the discrete ordered expansion~(\ref{dendr1}), which as discussed in the previous section can be expressed as the discrete analogue of Magnus expansion, and also in terms of a suitable pre-Lie algebra.
{We note that the expression within the bracket in~(\ref{alt}) is similar to expression~(\ref{dendr1}), but with the objects
${\mathbb P}_n$ being now replaced by ${\mathrm P}_n$, defined~(\ref{alt0}).} Such expressions appear, for instance, in $q$-deformed algebras and in the context of algebras emerging from set-theoretic solutions of the Yang--Baxter equation. These cases however will be discussed in detail elsewhere.
\end{proof}

\subsection{Gauge {transformations} and ``folding''}\label{section3.2}

 In this subsection, we consider two fundamental applications
associated with the derivations of the previous section.
The first application is the transformation of the linear operator~${\mathbb L}$ ($A$~in the continuous case) via a gauge transformation,
and the second is the construction of the evolution problem and its solution when ``folding'' is included (see Figure~\ref{fig2}).

\subsubsection{Gauge transformations} Let $T_n$ be a solution of the difference equation
(\ref{discrete}), where ${\mathbb L}_n(\alpha) \in \EEnd\bigl({\mathbb C}^{\mathcal N}\bigr)$ (recall Remark~\ref{class}). Let also
${\mathrm G}_n(\alpha) \in \EEnd\bigl({\mathbb C}^{\mathcal N}\bigr)$ and define \smash{$\hat T_n(\alpha):= {\mathrm G}_n(\alpha) T_n(\alpha)$}, ${\mathbb L}, {\mathrm G} \in \EEnd\bigl({\mathbb C}^{\mathcal N}\bigr)$ are invertible.
Then
$\hat T_n$ is a solution of the difference equation
\begin{equation}
\hat T_{n+1}(\alpha) = \hat {\mathbb L}_n(\alpha) \hat T_n(\alpha), \label{discrete1}
\end{equation}
 where
$\hat {\mathbb L}_n(\alpha) = {\mathrm G}_{n+1}(\alpha) {\mathbb L}_n(\alpha) {\mathrm G}_{n}^{-1}(\alpha)$,
and hence
\begin{equation}
{\mathrm G}_{n+1}(\alpha) = \hat {\mathbb L}_n(\alpha) {\mathrm G}_n(\alpha) {\mathbb L}^{-1}_n(\alpha).\label{gauge1}
\end{equation}
 Assuming that the operators ${\mathbb L_n}$, $\hat {\mathbb L}_n$ are known, we solve for ${\mathrm G}_n$.

\begin{lem}
Let $T_{N+1} (\alpha): = {\mathbb L}_{N}(\alpha)\cdots {\mathbb L}_{1}(\alpha)$ and $\hat T_{N+1} (\alpha): = \hat {\mathbb L}_{N}(\alpha)\dots \hat {\mathbb L}_{1}(\alpha)$, where $L_n$, $\hat L_n \in\EEnd\bigl({\mathbb C}^{\mathcal N}\bigr)$, be solutions of~\eqref{discrete} and~\eqref{discrete1} respectively. Then, the solution of the difference equation~\eqref{gauge1} is given by \smash{${\mathrm G}_{N+1}(\alpha) = \hat T_{N+1}(\alpha) {\mathrm G}_1(\alpha) T^{-1}_{N+1}(\alpha)$}, where ${\mathrm G}_1(\alpha)$ is some generic initial value.
\end{lem}
\begin{proof} The solution of~(\ref{gauge1}) is obtained directly by iteration.
\end{proof}

Similarly, in the continuous case, let $T(x,\alpha)$ be a solution of the time evolution problem~(\ref{evo1}), where $A(x) \in \EEnd\bigl({\mathbb C}^{\mathcal N}\bigr)$. Let also $G(x, \alpha)\in \EEnd\bigl({\mathbb C}^{\mathcal N}\bigr)$ and define
$\hat T(x, \alpha) := {\mathrm G}(x, \alpha) T(x, \alpha)$. Then $ \hat T$ satisfies
\begin{equation}
\partial_{\xi} \hat T(\xi, \alpha)= \alpha \hat A(\xi) \hat T(\xi, \alpha). \label{evo2}
\end{equation}
From~(\ref{evo1}) and~(\ref{evo2}) and $\hat T(x, \alpha) := {\mathrm G}(x, \alpha) T(x, \alpha)$, we conclude that the transformed operator~$\hat A(x) $ is given as \begin{equation*}
\hat A(\xi) = {\mathrm G}(\xi, \alpha)A{\mathrm G}^{-1} (\xi, \alpha)+\alpha^{-1}\partial_{\xi} {\mathrm G}(\xi, \alpha){\mathrm G}^{-1}(\xi, \alpha), \nonumber
\end{equation*}
which is a {typical gauge transformation of $A$.}
Suppose that $A$, $\hat A$ are given, then from the latter equation we conclude that
${\mathrm G}(\xi,\alpha)$ is a solution of the evolution problem
\begin{equation}
\partial_{\xi} {\mathrm G}(\xi, \alpha) = \alpha\hat A(\xi) {\mathrm G}(\xi, \alpha) - \alpha {\mathrm G}(\xi, \alpha) A(\xi).\label{gaugeA}
\end{equation}
The solution of the latter equation is given by
\[
{\mathrm G}(x, \alpha) = \hat T(x, \alpha) {\mathrm G}_0(\alpha)T^{-1}(x, \alpha),
\]
 where we recall that, \smash{$T(x, \alpha )= \overset{\curvearrowleft} {\mathrm P} \mbox{exp}\bigl( \alpha \int_{0}^{x} A(\xi){\rm d}\xi \bigr)$} (similarly for $\hat T(x, \alpha)$) and ${\mathrm G}_0(\alpha)$ is some initial value at~${x =0}$.

This type of problems systematically appear in the context of integrable systems, where a~Lax pair and strong compatibility conditions (zero curvature condition, i.e., equations of motion) exist~\cite{AKNS, FT, Lax}. In this frame, so-called dressing schemes~\cite{DS, MS, ZS} are used in order to obtain solutions of the associated integrable non-linear ODEs and PDEs that emerge from the zero curvature condition.

Typically, in integrable systems, due to the existence of a Lax pair or the existence of a~classical or quantum $R$-matrix~\cite{FT, FadTakRes} it turns out that $\operatorname{tr} T(\alpha)$ provides a hierarchy of conserved quantities; for instance, the Hamiltonian of the system with periodic boundary conditions is a~member of this hierarchy (a more detailed discussion will follow later in the text). In the next section, we are considering both quantum and classical integrable systems, and we establish fundamental connections with the findings of the present section. In what follows a relevant construction, which is associated with a time evolution problem that includes a ``folding'' is discussed.

\subsubsection{Time evolution with ``folding''}
We now focus on the evolution problem that involves a ``folding''. Such evolution problems are also associated with integrable systems with open boundary conditions (instead of periodic ones), although we are not going to discuss here the integrability conditions for such systems (the interested reader is referred to~\cite{Sklyanin, Sklyaninb} for a detailed exposition).

In order to describe discrete systems with a ``folding'', we also need to introduce the so called backward time evolution equation. In order to achieve this, we need to introduce the operator
\[
\hat {\mathbb L}(a)= 1 + \sum_{m>0} \alpha^m \hat L^{(m)} \in \EEnd\bigl({\mathbb C}^{\mathcal N}\bigr) \otimes {\mathfrak A},
\]
where $a\in {\mathbb C}$, $\hat L^{(m)}$ are linear operators
\[
\hat L^{(m)} = \sum_{q,b =1}^{{\mathcal N}} e_{a,b} \otimes \hat L_{a,b}^{(m)} \in \EEnd\bigl({\mathbb C}^{\mathcal N}\bigr) \otimes {\mathfrak A}
\]
 and ${\mathfrak A}$ is some unital, associative, ${\mathbb C}$-algebra with generators \smash{$\hat L_{a,b}^{(m)}$}, $a,b \in \{1,2 \dots ,{\mathcal N}\}$, $m \in {\mathbb Z}^+$ (recall also Remark~\ref{class}).
Then the {\it backward} time evolution equation is defined as
\begin{equation}
\hat T_{n+1}(\alpha) = \hat T_{n}(\alpha) \hat {\mathbb L}_n(\alpha). \label{discrete2}
\end{equation}
The solution of equation~(\ref{discrete2}) is found by iteration and is given by the backward monodromy matrix
\begin{equation}
\hat T_{N+1}(\alpha) = \hat {\mathbb L}_{1}(\alpha) \dots \hat {\mathbb L}_{N}(\alpha) \in \EEnd\bigl({\mathbb C}^{\mathcal N}\bigr) \otimes {\mathfrak A}^{\otimes N}, \label{back}
\end{equation}
(we choose for simplicity $\hat T_1(\alpha) = 1$ as initial condition).
We express the monodromy as
\[
\hat T_{N+1}(\alpha)=1 + \sum_{n\geq 1} \alpha^n \hat T^{(n)}(N+1)
\]
and recall the exponential map
\smash{$\hat T_{N+1}(\alpha) = {\rm e}^{\hat {\mathcal Q}_{N+1}(\alpha)}$};
{expressions}~(\ref{a00}),~(\ref{ac}),~(\ref{R1}) and (\ref{R2}) naturally hold.

We consider the general case, where \smash{$\hat {\mathbb L}(\alpha) = 1 + \sum_{k \geq 1} \alpha^k \hat L^{(k)}$}, then the form~of the coefficients of the monodromy \smash{$\hat T$} are given via~(\ref{back}) as
\begin{gather}
\hat T^{(m)}(N+1)= \sum_{\sum_{j=1}^km_j=m}\Biggl( \sum_{n_{k}=1}^{n_{k-1}-1 }\hat L_{n_k}^{(m_k)} \sum_{n_{k-1}=1}^{n_{k-2}-1 }\hat L^{(m_{k-1})}_{n_{k-1}} \cdots \sum_{n_{2}=1}^{n_1-1 }\hat L^{(m_2)}_{n_2}\sum_{n_1=1}^N \hat L^{(m_1)}_{n_1}\Biggr). \label{tt2b}
\end{gather}

We may then express the elements of the discrete analogue of Dyson's series~(\ref{tt2b}) in terms of the
tridendriform algebra in analogy to Proposition~\ref{lemma01}.
{\begin{lem}{\label{lemmahat}}
The elements of the discrete Dyson series~\eqref{tt2b} are expressed as
\begin{eqnarray}
\hat T^{(m)}(N+1) = \sum_{\sum_{j=1}^km_j=m}\sum_{n=1}^N \bigl( \bigl(\cdots \bigl( \hat L^{(m_k)}\succ \hat L^{(m_{k-1})}\bigr) \succ \cdots \bigr)\succ \hat L^{(m_1)}\bigr)_n , \label{dendr2b}
\end{eqnarray}
where the tridendriform operation $\succ$ is defined as
\smash{$(a\succ b)_n := \Sigma(a)_n b_n$}, where \smash{$\Sigma(a)_n\!=\! \sum_{m=1}^{n-1} a_m$}, $a, b \in \EEnd\bigl({\mathbb C}^{\mathcal N}\bigr) \otimes{\mathfrak A}$ $($see also Remark~$\ref{class}).$

For instance, for $n=1$,
\[
\hat T^{(1)}(N+1)= \sum_{n=1}^N \hat L^{(1)}_n,
\]
 for $n=2$ and $n=3$
\begin{gather*}
\hat T^{(2)}(N+1) =
 \sum_{n=1}^N \bigl( \bigl(\hat L^{(1)} \succ \hat L^{(1)}\bigr)_n + \hat L_n^{(2)}\bigr),\\
\hat T^{(3)}(N+1)= \sum_{n=1}^N \bigl( \bigl(\bigl(\hat L^{(1)} \succ \hat L^{(1)}\bigr)\succ \hat L^{(1)}\bigr)_n +
\bigl(\hat L^{(1)}\succ L^{(2)} \bigr)_n + \bigl(\hat L^{(2)} \succ \hat L^{(1)}\bigr)_n + \hat L_n^{(3)}\bigr).
\end{gather*}
\end{lem}
\begin{proof} The proof goes along the same lines as in Proposition~\ref{lemma01} with $(a\succ b)_n := \Sigma(a)_n b_n$.
\end{proof}}

We recall the exponential map \smash{$\hat T(\alpha) = {\rm e}^{\hat {\mathcal Q}(\alpha)}$},
\smash{$\hat {\mathcal Q}(\alpha) = \sum_{m=1}^{\infty}\alpha^m \hat {\mathcal Q}^{(m)}$},
which will lead to the discrete analogue of Magnus expansion.
{Recursive expressions~(\ref{R1}) and~(\ref{R2}) hold for~\smash{$\hat T^{(m)}$},~\smash{$\hat {\mathcal Q}^{{m}}$}, but now $\hat T^{(m)}$ are given by~(\ref{dendr2b}).}

\begin{lem}
Let
\[
\hat T_{N+1}(\alpha) = {\rm e}^{\hat {\mathcal Q}(\alpha)}, \qquad \hat T_{N+1}(\alpha)=1+\sum_{m=1}^{\infty}\alpha^m \hat T^{(m)}(N+1),
\]
 where $\hat T^{(m)}(N+1)$ are given in Lemma~$\ref{lemmahat}$ and
\smash{$\hat {\mathcal Q}(\alpha) = \sum_{m=1}^{\infty}\alpha^m \hat {\mathcal Q}^{(m)}$}. The quantities
\[
\hat {\mathcal Q}^{(k)}(N+1):= \sum_{n=1}^N \hat \Omega^{(k)}_n,
\]
where \smash{$\hat \Omega^{(k)}_n = \hat {\mathcal Q}^{(k)}(n+1) - \hat {\mathcal Q}^{(k)}(n)$},
are expressed explicitly as
\begin{gather}
\hat {\mathcal Q}^{(1)}(N+1) = \sum_{n=1}^{N} \hat L^{(1)}_{n},\cr
\hat {\mathcal Q}^{(2)}(N+1) = \frac{1}{2}\sum_{n>n_1=1}^{N} \bigl[ \hat L^{(1)}_{n_1}, L^{(1)}_{n} \bigr] -
\frac{1}{2}\sum_{n=1}^{N} \bigl(\hat L^{(1)}_{n}\bigr)^2+\sum_{n=1}^N \hat L_n^{(2)}, \cr
\hat {\mathcal Q}^{(3)}(N+1) = \sum_{n=1}^N \biggl(\frac{1}{6} \sum_{m>k=1}^{n-1}
\bigl(\bigl[\hat L^{(1)}_{k}, \bigl[\hat L^{(1)}_{m}, \hat L^{(1)}_{n} \bigr] \bigr] +
\bigl[\bigl[\hat L^{(1)}_{k}, \hat L^{(1)}_{m} \bigr], \hat L^{(1)}_{n} \bigr]\bigr) \cr
\hphantom{\hat {\mathcal Q}^{(3)}(N+1) =\sum_{n=1}^N \biggl(}{}
+ \frac{1}{6} \sum_{m=1}^{n-1} \bigl(\hat L^{(1)}_{n}\bigl[\hat L^{(1)}_{n},
\hat L^{(1)}_{m}\bigr] + \bigl[\hat L^{(1)}_{n}, \hat L^{(1)}_{m}\bigr] \hat L^{(1)}_{m}\bigr) \cr
\hphantom{\hat {\mathcal Q}^{(3)}(N+1) =\sum_{n=1}^N \biggl(}{}
+\frac{1}{6} \sum_{m=1}^{n-1} \bigl(\bigl[\hat L^{(1)}_{n}, (\hat L^{(1)}_{m})^2\bigr] +
\bigl[\bigl(\hat L^{(1)}_{n}\bigr)^2, \hat L^{(1)}_{m}\bigr] \bigr)\nonumber\\
\hphantom{\hat {\mathcal Q}^{(3)}(N+1) =\sum_{n=1}^N \biggl(}{}
-\frac{1}{2} \sum_{n=1}^{N} \bigl(\hat L_n^{(1)} \hat L_n^{(2)} +\hat L_n^{(2)}\hat L_n^{(1)} \bigr)+\frac{1}{3} \bigl(\hat L^{(1)}_{n}\bigr)^3 \nonumber\\
\hphantom{\hat {\mathcal Q}^{(3)}(N+1) =\sum_{n=1}^N \biggl(}{}
-\frac{1}{2} \sum_{m=1}^{n-1}\bigl( \bigl[\hat L^{(1)}_n, \hat L_m^{(2)} \bigr] +\bigl[\hat L^{(2)}_n \hat L^{(1)}_m \bigr] \bigr) + \hat L_n^{(3)} \biggr), \qquad \dots\,. \label{QQb}
\end{gather}
\end{lem}
\begin{proof}
The proof is along the lines of Lemma~\ref{lemma2}.
\end{proof}

We are now in the position to express the discrete Magnus
expansion elements is terms of the right pre-Lie operation $\triangleleft$.
Let $A$, $B$ be linear operators that depend on a discrete parameter~${n\in {\mathbb N}}$ (see Remark~\ref{class}))
and define the binary operation
$ \triangleleft\colon (A, B) \mapsto A \triangleleft B$ such that
\begin{equation}
(A\triangleleft B)_n := \Biggl[A_n, \sum_{m=1}^{n-1} B_m \Biggr] +A_n B_n. \label{bina2b}
\end{equation}
Then the right pre-Lie identity is satisfied, i.e.,
\begin{equation}
( (A \triangleleft B) \triangleleft C)_n- ( A \triangleleft (B \triangleleft C))_n =( (A \triangleleft C) \triangleleft B)_n- (A \triangleleft(C \triangleleft B))_n. \label{sym2}
\end{equation}

\begin{pro}
The elements of the discrete Magnus expansion~\eqref{QQb} can be re-expressed in terms of the pre-Lie operation~\eqref{bina2b}.
\end{pro}
\begin{proof}
Recalling from~(\ref{bina2b}) and~(\ref{sym2}) that
\begin{equation}
(a \triangleleft b)_n = \Biggl[a_n, \sum_{m=1}^{n-1}b_m\Biggr] + a_n b_n, \label{tria1b}
\end{equation}
$a,b \in \EEnd\bigl({\mathbb C}^{\mathcal N}\bigr) \otimes {\mathfrak A}$, we conclude
\begin{gather*}
 \hat {\mathcal Q}^{(1)}(N+1) = \sum_{n=1}^N\hat L^{(1)}_n, \cr
 \hat {\mathcal Q}^{(2)}(N+1) =-\frac{1}{2} \sum_{n=1}^N
 \bigl(\hat L^{(1)} \triangleleft \hat L^{(1)}\bigr)_n +\sum_{n=1}^N\hat L_n^{(2)},\nonumber\\
 \hat {\mathcal Q}^{(3)}(N+1)= \sum_{n=1}^N\biggl(
\frac{1}{12} \bigl(\bigl(\hat L^{(1)} \triangleleft \hat L^{(1)}\bigr) \triangleleft \hat L^{(1)}\bigr)_n +\frac{1}{4}
\bigl(\hat L^{(1)} \triangleleft \bigl(\hat L^{(1)} \triangleleft \hat L^{(1)}\bigr)\bigr)_n \biggr)\cr
\hphantom{\hat {\mathcal Q}^{(3)}(N+1)= \sum_{n=1}^N\biggl(}{}
 -\frac{1}{2} \sum_{n=1}^N\bigl(\bigl(\hat L^{(2)} \triangleleft L^{(1)}\bigr)_n+ \bigl(\hat L^{(1)}\triangleleft \hat L^{(2)}\bigr)_n\bigr)
+ \sum_{n=1}^N \hat L_n^{(3)}, \qquad \dots\,. \nonumber
\end{gather*}

As is Proposition~\ref{prom}, we only compute the first few terms of the expansion. Using the recursive relations~(\ref{R1}) and~(\ref{R2}), expressions~(\ref{tt2b}) and~(\ref{QQb}) and definition~(\ref{tria1b}), we may explicitly compute higher-order terms, however such computations becomes increasingly complicated due to the existence of extra terms (see for instance the last term in \smash{$\hat {\mathcal Q}^{(2)}$}, the last two terms in line three \smash{$\hat {\mathcal Q}^{(3)}$} and the last three terns in the last line of \smash{$\hat {\mathcal Q}^{(3)}$,~(\ref{QQb}))}.
\end{proof}

A remark similar to Remark~\ref{remark12} applies here too by setting
$\hat L^{(1)} = \hat {\mathbb P}$, $\hat L^{(j)} = 0$, $\forall j>1$.

Now that we have the discrete Magnus expansions for both the
forward and backward evolution problems we may identify the solution
of evolution problem with ``folding''~\cite{Sklyanin, Sklyaninb}.

\begin{lem} 
Let ${\mathbb L}, \hat {\mathbb L} \in \EEnd\bigl({\mathbb C}^{\mathcal N}\bigr) \otimes {\mathfrak A}$ and
\[
T_{N+1}(\alpha) = {\mathbb L}_N(\alpha)\cdots {\mathbb L}_1(\alpha), \qquad \hat T_{N+1}(\alpha) = \hat {\mathbb L}_1(\alpha) \cdots \hat {\mathbb L}_N(\alpha)
\]
be solutions of \eqref{discrete} and~\eqref{discrete2}, respectively. Let also $K(\alpha) \in \EEnd\bigl({\mathbb C}^{\mathcal N}\bigr)$, then the quantity
${\mathbb T}_{N+1}(\alpha)= T_{N+1}(\alpha) K(\alpha) \hat T_{N+1}(\alpha)$~{\rm \cite{Sklyanin, Sklyaninb}},
is a solution of the difference equation
\begin{equation}
{\mathbb T}_{n+1}(\alpha) = {\mathbb L}_{n}(\alpha) {\mathbb T}_n(\alpha)
\hat {\mathbb L}_{n}(\alpha). \label{open}
\end{equation}
\end{lem}
\begin{proof}
The proof is straightforward by means of~(\ref{discrete}) and~(\ref{discrete2}).
\end{proof}

\begin{figure}[!ht]\centering
\vspace{10mm}

\begin{picture}(1,1)
	\put(-47,20){\line(1,0){130}}
\textcolor{blue}{\put(-47,-20){\line(1,0){130}}}
	\multiput(-39,30)(0,-8){8}{\line(0,-1){6}}
	\multiput(-26,30)(0,-8){8}{\line(0,-1){6}}
	\multiput(76,30)(0,-8){8}{\line(0,-1){6}}
	\put(16,-3){$\cdots\cdots$}
	\put(83,20){\line(1,-1){20}}
\textcolor{blue}{\put(83,-20){\line(1,1){20}}}
	\put(-44,-50){$N$}
	\put(-28,-50){$N-1$}
	\put(74,-50){$1$}
 \put(10, 25){$\longrightarrow$}
 \put(10, -35){${\mathbf \longleftarrow}$}
 \put(110, 0){$K$}
	\put(-80,20){$T_{N+1}$}
 \put(-80,-20){$\hat T_{N+1}$}\end{picture}

\vspace{15mm}

\caption{Graphical representation of ${\mathbb T}_{N+1}$ (``folding'').}\label{fig2}
\end{figure}

Figure~\ref{fig2} reads from left top to right bottom clockwise (following the arrows).

\begin{rem}
We first recall that
${\mathbb L}(\alpha)$, $\hat {\mathbb L}(\alpha)$ belong to \smash{$\EEnd\bigl({\mathbb C}^{\mathcal N}\bigr) \otimes {\mathfrak A}$}. We consider $\hat {\mathbb L}(\alpha) := {\mathbb L}^{-1}(-\alpha)$
\big(where \smash{${\mathbb L}(\alpha) {\mathbb L}^{-1}(\alpha) = 1_{{\mathbb C}_{\mathcal N}} \otimes 1_{\mathfrak A}$}\big) and consequently \smash{$\hat T_{N+1}(\alpha) = T_{N+1}^{-1}(-\alpha)$}. This case is associated with the so-called {\it reflection algebra} in integrable systems~\cite{Cherednik, Sklyaninb}.
In the continuum limit, after recalling ${\mathbb T}_{n+1} \to {\mathbb T}(x + \delta)$, ${\mathbb L}_n \to 1 + \delta\alpha A(x)$ and keeping linear terms we obtain from~(\ref{open})
\begin{equation}
\partial_{\xi}{\mathbb T}(\xi; \alpha) = \alpha A(\xi){ \mathbb T}(\xi; \alpha) + \alpha { \mathbb T}(\xi; \alpha) A(\xi), \label{openb}
\end{equation}
which is the evolution equation that includes ``folding''. The solution of~(\ref{openb}) is given by
\[
{\mathbb T}(\xi, \alpha) = T(\xi, \alpha) K(\alpha) \hat T(\xi, \alpha),
\]
 where \smash{$T(\xi, \alpha) = {\rm e}^{{\mathcal Q}(\xi, \alpha)}$} is given in~(\ref{eq:CTS_T}), and \smash{$\hat T(\xi ,\alpha) = T^{-1}(\xi,- \alpha)$}.
 Notice that although~(\ref{gaugeA}) and~(\ref{openb}) are mathematically similar, they have distinct physical interpretations as already pointed out.
\end{rem}

By means of a suitable combination of solutions of the forward and backward linear evolution problem and the presence of $K$ (Figure~\ref{fig2} above), integrable systems with generic boundary conditions can be constructed~\cite{Sklyanin, Sklyaninb}. Specifically, the trace of ${\mathbb T}_{N+1}$ produces a hierarchy of conserved quantities (Hamiltonians) that describe integrable systems with general open boundary conditions (as opposed to periodic). From the point of view of integrable systems and quantum algebras, the $K$-matrix and the object ${\mathbb T}$ are representations of the reflection algebra~\cite{Cherednik, Sklyanin, Sklyaninb}, which is a left (right) coideal of the underlying quantum (deformed) algebra of the integrable system at hand.

\section{Quantum groups as tridendriform and pre-Lie algebras}\label{section4}

 After the derivation of the discrete analogue of the Magnus expansion and the discrete pre-Lie expansion, we are ready to study some key relationships between quantum groups, tridendriform and pre-Lie algebras. We employ the Yangian ${\mathcal Y}(\mathfrak{gl}_{\mathcal N})$ (sometimes we simply write ${\mathcal Y}$ in the manuscript)~\cite{Chari, Drinfeld, Drinfeld1, Drinfeld2, Molev} as our key paradigm and show that the $N$-coproducts of the Yangian elements can be re-expressed in terms of suitable tridendriform and pre-Lie algebra binary operations. A relevant interesting work on the homomorphism between quasi-shuffle algebras and the Yangian can be found in~\cite{Hudson}.
Our approach here is direct and is based on the Faddeev--Reshetikhin--Takhtajan (FRT) construction~\cite{FadTakRes} and the explicit derivation of $N$-coproducts of the Yangian via tensorial representations of the associated FRT algebra as will be clear in the following sections.

We first recall the derivation of quantum groups (or quantum algebras) associated to solutions $R\colon V \otimes V \to V \otimes V $ of the Yang--Baxter equation (YBE)~\cite{Baxter, Yang} \big(henceforth $V = {\mathbb C}^{{\mathcal N}}$\big)
\begin{equation}
R_{12}(\lambda_1, \lambda_2)R_{13}(\lambda_1, \lambda_3)R_{23}(\lambda_2, \lambda_3)= R_{23}(\lambda_2, \lambda_3) R_{13}(\lambda_1, \lambda_3) R_{12}(\lambda_1,\lambda_2), \label{frt}
\end{equation}
where $\lambda_1, \lambda_2 \in {\mathbb C}$. Let \smash{$R = \sum_ x {\mathrm a}_x \otimes {\mathrm b}_x$}, where \smash{${\mathrm a}_x, {\mathrm b}_x\in \EEnd\bigl({\mathbb C}^{\mathcal N}\bigr)$}, then in the ``index notation'' \smash{$ R_{12} = \sum_x {\mathrm a}_x\otimes {\mathrm b}_x \otimes 1_V$}, \smash{$ R_{23} = 1_{V} \otimes \sum_x {\mathrm a}_x \otimes {\mathrm b}_x$}, and \smash{$R_{13} = \sum_x {\mathrm a}_x \otimes 1_V \otimes {\mathrm b}_x$}.

For the derivation of a quantum algebra associated to a given $R$-matrix, we employ the FRT construction.
We recall the standard ${\mathcal N} \times {\mathcal N}$ matrices $e_{x,y}$, with entries $(e_{x,y})_{z,w} = \delta_{x,z} \delta_{y,w}$, $x,y,z,w \in \{1,2, \dots, {\mathcal N}\}$.
\begin{defn} \label{frt1}
Let $R(\lambda_1, \lambda_2)\in \EEnd(V \otimes V)$ be a solution of the Yang--Baxter equation~(\ref{frt}), $\lambda_1, \lambda_2 \in {\mathbb C}$, $V ={\mathbb C}^{\mathcal N}$. Let also \smash{${\mathbb L}(\lambda) := \sum_{x,y\in 1}^{\mathcal N} e_{x,y} \otimes L_{x,y}(\lambda) \in \EEnd(V) \otimes {\mathfrak A}$}, where $\lambda\in {\mathbb C}$ and \smash{$L_{x,y}(\lambda) =\sum_{m=0}^{\infty}\lambda^{-m} L^{(m)}_{x,y}\in {\mathfrak A}$}. The quantum algebra ${\mathfrak A}$, associated to $R$, is defined as the quotient of the free unital, associative ${\mathbb C}$-algebra, generated by indeterminates
\[
\bigl\{ L^{(m)}_{x,y}\mid x,y \in \{1,2, \dots,{\mathcal N}\},\, m\in\{0,1,2, \dots\}\bigr\}\]
and relations
\begin{equation}
R_{12}(\lambda_1, \lambda_2){\mathbb L}_1(\lambda_1){\mathbb L}_2(\lambda_2) = {\mathbb L}_2(\lambda_2)\
{\mathbb L}_1(\lambda_1)R_{12}(\lambda_1, \lambda_2), \label{RTT}
\end{equation}
where $R_{12} =R \otimes 1_{\mathfrak A }$ and
\smash{${\mathbb L}_{1}=\sum_{x,y\in 1 }^{\mathcal N}e_{x,y}\otimes 1_V \otimes L_{x,y}$},\footnote{Notice that in ${\mathbb L}$ in addition to the indices 1 and 2 in~(\ref{RTT}) there is also an implicit ``quantum index'' $3$ associated to ${\mathfrak A}$,
which for now is omitted, i.e., one writes ${\mathbb L}_{13}$, ${\mathbb L}_{23}$.} \smash{${\mathbb L}_{2}=\sum_{x,y \in 1}^{\mathcal N} 1_V \otimes e_{x,y}\otimes L_{x,y}$}.
\end{defn}
It is worth noting that if equation~(\ref{RTT}) holds, then $R$ is necessarily a solution of the Yang--Baxter equation~(\ref{frt}) (see, e.g.,~\cite{Majid} for a proof). Definition~\ref{frt} basically states that different choices of solutions of the Yang--Baxter equation yield distinct quantum algebras.

\subsection[Yangian Y(gl\_N)]{Yangian $\boldsymbol{{\mathcal Y}(\mathfrak{gl}_{\mathcal N})}$}

We present in this subsection a concise review of the $\mathfrak{gl}_{\mathcal N}$ Yangian
${\mathcal Y}(\mathfrak{gl}_{\mathcal N})$, which is a special example of a quantum algebra, and provide all the necessary for our analysis information (for a~more detailed exposition the interested reader is refereed for instance to~\cite{Chari, Molev}). We consider the FRT point of view (Definition~\ref{frt1}). We set $\alpha^{-1}=\lambda$ (additive parameter) and consider~${R(\lambda_1, \lambda_2)= R(\lambda_1 - \lambda_2)}$. Specifically, in the case of the $\mathfrak{gl}_{\mathcal N}$ Yangian (${\mathcal Y}$ for brevity), the $R$-matrix is given by, $R(\lambda_1, \lambda_2) = (\lambda_1 - \lambda_2) 1_{V \otimes V} + {\mathcal P}$, where \smash{${\mathcal P} = \sum_{i,j =1}^{\mathcal N} e_{i,j} \otimes e_{j,i}$} is the permutation operator such that ${\mathcal P}(a\otimes b) = b \otimes a$, $a,b \in V$ and
\[
{\mathbb L}(\lambda) = 1_{V \otimes {\mathcal Y}} + \sum_{k=1}^{\infty} \lambda^{-k} L^{(m)}, \qquad
L^{(m)} = \sum_{x,y=1}^{\mathcal N} e_{x,y} \otimes L_{x,y}^{(m)}.
\]
 Then, by the fundamental relation~(\ref{RTT}), the algebraic relations among the generators \smash{$L_{x,y}^{(m)}$} of the $\mathfrak{gl}_{\mathcal N}$ Yangian are deduced and are given in the following definition.
\begin{defn}
The $\mathfrak{gl}_{\mathcal N}$ Yangian ${\mathcal Y}(\mathfrak{gl}_{\mathcal N})$ is a unital, associative algebra generated by indetermi\-nates $1_{\mathcal Y}$ (unit element) and \smash{$L_{i,j}^{(m)}$}, $i,j \in \{1,2,\dots,{\mathcal N}\}$, $m \in \{0,1,2,\dots\}$ and relations
\begin{equation}
\bigl[ L_{i,j}^{(p+1)}, L_{k,l}^{(m)}\bigr] -\bigl[ L_{i,j}^{(p)}, L_{k,l}^{(m+1)}\bigr] = L_{k,j}^{(m)}L_{i,l}^{(p)}- L_{k,j}^{(p)}L_{i,l}^{(m)}, \label{fund2b}
\end{equation}
where $[ \ ,\, ]\colon {\mathcal Y}(\mathfrak{gl}_{\mathcal N})\times {\mathcal Y}(\mathfrak{gl}_{\mathcal N})\to {\mathcal Y}(\mathfrak{gl}_{\mathcal N})$, such that $[a, b] = ab - ba$, for all $a,b \in {\mathcal Y}(\mathfrak{gl}_{\mathcal N})$.
\end{defn}

We recall the Lie exponential is written as \smash{${\rm e}^{\mathrm Q} := 1+ \sum_{n=1}^{\infty}\frac{{\mathrm Q}^n}{n!}$}.
We focus here on the case were ${\mathbb L}, {\mathrm Q} \in \EEnd\bigl({\mathbb C}^{\mathcal N}\bigr)
\otimes {\mathcal Y}(\mathfrak{gl}_{\mathcal N})$, and ${\mathbb L}$ satisfies relation~(\ref{RTT}).
We express the generic solution of~(\ref{RTT}) as \smash{${\mathbb L}(\lambda) = {\rm e}^{{\mathrm Q}(\lambda)}$}
and consider the formal $\lambda$ series expansions
\begin{equation}
{\mathbb L}(\lambda) = 1_{V \otimes {\mathcal Y}}+ \sum_{m=1}^{\infty}\lambda^{-m} L^{(m)},
\qquad {\mathrm Q}(\lambda) = \sum_{m=1}^{\infty}\lambda^{-m} {\mathrm Q}^{(m)}. \label{lexp}
\end{equation}
Then comparing the series expansion ${\mathbb L}(\lambda)$ and \smash{${\rm e}^{{\mathrm Q}(\lambda)}$}, using also~(\ref{lexp}),
we obtain expressions of~${\mathrm Q}^{(m)}$ in terms of symmetric polynomials of $L^{(m)}$ \big(see~(\ref{a00}) and~(\ref{ac}), ${{\mathcal Q}^{(n)}\! \to\! {\mathrm Q}^{(n)}}$, ${T^{(n)}\! \to\! L^{(n)} }$\big).
\begin{gather}
{\mathrm Q}^{(1)} = L^{(1)}, \qquad {\mathrm Q}^{(2)} = L^{(2)} -\frac{1}{2} \bigl(L^{(1)}\bigr)^2, \label{a0} \\
{\mathrm Q}^{(3)} = L^{(3)} - \frac{1}{2} \bigl( L^{(1)} L^{(2)} + L^{(2)} L^{(1)}\bigr) +\frac{1}{3} \bigl(L^{(1)}\bigr)^3, \qquad\dots, \label{acd}
\end{gather}
and vice versa as in~(\ref{a}) and~(\ref{c}), i.e., similarly, the logarithm can be defined such that $\ln{{\mathbb L}(\lambda)}= {\mathrm Q}(\lambda)$.

Our aim now is to derive an alternative set of generators of the Yangian
based on expressions~(\ref{a0}) and~(\ref{acd}).
Indeed, let us focus in the first few explicit exchange relations from~(\ref{fund2b}):
\begin{enumerate}\itemsep=0pt
\item $n=0$, $m=1$ $\bigl(L_{i,j}^{(0)} = \delta_{i,j}\bigr)$,
\begin{equation*}
\bigl[ L_{i,j}^{(1)}, L_{k,l}^{(1)}\bigr] = \delta_{i,l} L_{k,j}^{(1)} -\delta_{k,j} L_{i,l}^{(1)},
\end{equation*}
the latter are the familiar ${\mathfrak gl}_{\mathcal N}$ exchange relations.
\item $n=2$, $m=0$,
\begin{equation*}
\bigl[L_{i,j}^{(2)}, L_{k,l}^{(1)}\bigr] =\delta_{i,l} L_{k,j}^{(2)} - \delta_{k,j} L_{i,l}^{(2)}.
\end{equation*}

\item $n=2$, $m=1$,
\begin{equation*}
\bigl[ L_{i,j}^{(3)}, L_{k,l}^{(1)}\bigr] -\bigl[ L_{i,j}^{(2)}, L_{k,l}^{(2)}\bigr] =
L_{k,j}^{(1)} L_{i,l}^{(2)} - L_{k,j}^{(2)} L_{i,l}^{(1)}.
\end{equation*}

\item $n=3$, $m=0$,
\begin{equation*}
\bigl[L_{i,j}^{(3)}, L_{k,l}^{(1)}\bigr] =\delta_{i,l} L_{k,j}^{(3)} - \delta_{k,j} L_{i,l}^{(3)}.
\end{equation*}
\end{enumerate}

\begin{lem}\label{lemmaq}
The algebraic relations for the alternative set of generators of the Yangian ${\mathcal Y}(\mathfrak{gl}_{\mathcal N})$,
${\mathrm Q}_{i,j}^{(m)}$, $i,j \in \{1, \dots, {\mathcal N} \}$, $m \in {\mathbb Z}^+$ are given
\begin{gather*}
\bigl[ {\mathrm Q}_{i,j}^{(1)}, {\mathrm Q}_{k,l}^{(1)} \bigr] =\delta_{i,l}{\mathrm Q}^{(1)}_{k,j}- \delta_{k,j}{\mathrm Q}^{(1)}_{i,l},\\
 \bigl[ {\mathrm Q}_{i,j}^{(1)}, {\mathrm Q}_{k,l}^{(2)} \bigr] =\delta_{i,l}{\mathrm Q}^{(2)}_{k,j}-\delta_{k,j}
{\mathrm Q}^{(2)}_{i,l},\\
\bigl[ {\mathrm Q}_{i,j}^{(2)}, {\mathrm Q}_{k,l}^{(2)} \bigr]
= \delta_{i,l}{\mathrm Q}^{(3)}_{k,j}- \delta_{k,j}{\mathrm Q}^{(3)}_{i,l}
-\frac{1}{4}{\mathrm Q}_{k,j}^{(1)}\sum_{x=1}^{\mathcal N} {\mathrm
Q}_{i,x}^{(1)}{\mathrm Q}_{x,l}^{(1)}+ \frac{1}{4}\sum_{x=1}^{\mathcal N}{\mathrm Q}_{k,x}^{(1)}{\mathrm Q}_{x,j}^{(1)}
{\mathrm Q}_{i,l}^{(1)}
\\
\hphantom{\bigl[ {\mathrm Q}_{i,j}^{(2)}, {\mathrm Q}_{k,l}^{(2)} \bigr]=}{}
 + \frac{1}{12} \Biggl( \delta_{k,j} \sum_{x,y=1}^{\mathcal N}{\mathrm Q}_{i,x}^{(1)}{\mathrm Q}_{x,y}^{(1)}
{\mathrm Q}_{y,l}^{(1)} - \delta_{i,l} \sum_{x,y=1}^{\mathcal N}{\mathrm Q}_{k,x}^{(1)} {\mathrm Q}_{x,y}^{(1)}
{\mathrm Q}_{y,j}^{(1)}\Biggr), \qquad \dots\,. 
\end{gather*}
\end{lem}
\begin{proof}
The proof is based on~(\ref{acd}) and the exchange relations (1)--(4).
\end{proof}

We also recall the definition of a Hopf algebra (see, for example,~\cite{Chari, Majid}) and then present the Yangian as a Hopf algebra.

\begin{defn} A Hopf algebra $(\mathcal{A}, \Delta, \epsilon, s)$ is a unital, associative algebra $\mathcal{A}$ over some field $k$ equipped with the following linear maps:
\begin{itemize}\itemsep=0pt
\item multiplication, $m\colon \mathcal{A} \otimes \mathcal{A} \to \mathcal{A}$, $m(a,b) = ab$, which is associative $(ab)c = a (bc)$ for all $a,b,c \in \mathcal{A}$,
\item $\eta\colon k \to \mathcal{A}$, such that it produces the unit element for $\mathcal{A}$, $\eta(1) = 1_\mathcal{A}$,
\item coproduct, $\Delta\colon \mathcal{A} \to \mathcal{A} \otimes \mathcal{A}$, $\Delta(a) = \sum_j \alpha_j \otimes \beta_j$, which is coassociative, $(\id\otimes \Delta)\Delta(a) = (\Delta \otimes \id)\Delta(a)$ for all $a\in \mathcal{A}$,
\item counit, $\epsilon\colon \mathcal{A} \to k$, such that $(\epsilon \otimes \id)\Delta(a) = (\id \otimes \epsilon)\Delta(a)= a$ for all $a \in \mathcal{A}$,
\item antipode, $s\colon \mathcal{A} \to \mathcal{A}$ (a bijective anti-algebra map), $m (s \otimes \id)\Delta(a) = m(\id \otimes s) \Delta(a) = \epsilon(a) 1_\mathcal{A}$ for all $a \in \mathcal{A}$.
\item $\Delta$, $\epsilon$ are algebra homomorphisms and $\mathcal{A} \otimes \mathcal{A}$ is endowed with its
usual tensor product algebra structure: $(a \otimes b)(c \otimes d) = ac\otimes bd$ for all $a,b,c,d \in \mathcal{A}$.
\end{itemize}
\end{defn}

It is also useful to introduce the definition of a quasi-triangular Hopf algebra~\cite{Drinfeld, Drinfeld1, Drinfeld2}, which is the analogue of the FRT construction.
\begin{defn} \label{quasi}
Let $\mathcal{A}$ be a Hopf algebra over some field $k$, then $\mathcal{A}$ is a quasi-triangular Hopf algebra if there exists an invertible element ${\mathcal R} \in \mathcal{A}\otimes \mathcal{A}$ (universal ${\mathcal R}$-matrix):
\begin{enumerate}\itemsep=0pt
 \item ${\mathcal R} \Delta(a) =\Delta^{\rm op}(a) {\mathcal R}$ for all $a \in \mathcal{A}$,
 where $\Delta\colon \mathcal{A} \to \mathcal{A} \otimes \mathcal{A}$ is the coproduct on $\mathcal{A}$ and $\Delta^{\rm op}(a) = \pi \circ \Delta(a)$, $\pi\colon \mathcal{A} \otimes \mathcal{A} \to \mathcal{A} \otimes \mathcal{A}$, such that $\pi (a \otimes b) = b \otimes a$.
 \item $(\id \otimes \Delta){\mathcal R} = {\mathcal R}_{13} {\mathcal R}_{12}$, and $(\Delta \otimes \id){\mathcal R} = {\mathcal R}_{13} {\mathcal R}_{23}$.
\end{enumerate}
\end{defn}
Also, the following statements hold~\cite{Majid}:
\begin{enumerate}\itemsep=0pt
\item[{(a)}] The antipode
$s\colon \mathcal{A} \to \mathcal{A}$
satisfies $(\id \otimes s){\mathcal R}^{-1} = {\mathcal R}$,
$(s \otimes \id){\mathcal R} ={\mathcal R}^{-1}$.
\item[{(b)}] The counit $\epsilon\colon \mathcal{A} \to k$ satisfies $(\id \otimes \epsilon) {\mathcal R} = (\epsilon \otimes \id){\mathcal R} = 1_\mathcal{A}$.
\item[{(c)}] Due to (1) and (2) of Definition~\ref{quasi} the universal ${\mathcal R}$-matrix satisfies the Yang--Baxter equation
\begin{equation}
{\mathcal R}_{12} {\mathcal R}_{13} {\mathcal R}_{23} ={\mathcal R}_{23} {\mathcal R}_{13} {\mathcal R}_{12}. \label{UYBE}
\end{equation}
\end{enumerate}
We recall the index notation: let
${\mathcal R} = \sum_j a_j \otimes b_j \in \mathcal{A} \otimes \mathcal{A}$, then
\[
{\mathcal R}_{12} = \sum_j a_j \otimes b_j \otimes 1_\mathcal{A}, \qquad
{\mathcal R}_{23} = \sum_j 1_\mathcal{A} \otimes a_j \otimes b_j,\qquad
{\mathcal R}_{13} = \sum_j a_j \otimes 1_\mathcal{A} \otimes b_j.
\]
Proofs of the above statements can be found in~\cite{Chari, Majid}.

\begin{rem} \label{remy} Consider a representation $\rho_{\lambda}\colon \mathcal{A} \to \EEnd\bigl({\mathbb C}^{\mathcal N}\bigr)$, $\lambda \in {\mathbb C}$, such that
\begin{gather*}
(\rho_{\lambda} \otimes \id){\mathcal R} =: {\mathbb L}(\lambda) \in \EEnd\bigl({\mathbb C}^{\mathcal N}\bigr) \otimes \mathcal{A},\\
(\rho_{\lambda_1} \otimes \rho_{\lambda_2}){\mathcal R} =: R(\lambda_1, \lambda_2)\in \EEnd\bigl({\mathbb C}^{\mathcal N}\bigr) \otimes \EEnd\bigl({\mathbb C}^{\mathcal N}\bigr), \qquad \lambda_{1,2} \in {\mathbb C}.
\end{gather*}
 Then the Yang--Baxter equation reduces to~(\ref{RTT}), after acting with $(\rho_{\lambda_1} \otimes \rho_{\lambda_2} \otimes \id)$ on~(\ref{UYBE}) and to the Yang--Baxter equation~(\ref{frt}), after acting with $(\rho_{\lambda_1} \otimes \rho_{\lambda_2} \otimes \rho_{\lambda_3})$ on~(\ref{UYBE}).
Also, from relations~(2) of Definition~\ref{quasi} and (a)--(b) above, we deduce
\begin{enumerate}\itemsep=0pt
 \item[{(i)}] The coproduct $\Delta\colon \mathcal{A} \to \mathcal{A} \otimes \mathcal{A}$ satisfies $(\id \otimes \Delta){\mathbb L}(\lambda) ={\mathbb L}_{13}(\lambda) {\mathbb L}_{12}(\lambda)$.
\item[{(ii)}] The counit $\epsilon\colon \mathcal{A} \to {\mathbb C}$ satisfies $(\id \otimes \epsilon) {\mathbb L}(\lambda) = 1_{V}$.
\item[{(iii)}] The antipode $s\colon \mathcal{A} \to \mathcal{A}$ satisfies $(\id \otimes s){\mathbb L}^{-1}(\lambda) ={\mathbb L}(\lambda)$.
\end{enumerate}
\end{rem}

{\bf The $\boldsymbol{\mathfrak{gl}_{\mathcal N}}$ Yangian as a quasi-triangular Hopf algebra.} In the case of the Yangian in particular, recall that
\[
{\mathbb L}(\lambda) = \sum_{m=0}^{\infty} \lambda^{-m } L^{(m)}= \sum_{m=0}^{\infty }\sum_{a,b =1}^{\mathcal N} \lambda^{-m}e_{a,b} \otimes L_{a,b}^{(m)}, \qquad L^{(0)}_{a,b} = \delta_{a,b}1_{\mathcal Y}.
\]
The~$\mathfrak{gl}_{\mathcal N}$ Yangian is a quasi-triangular Hopf algebra on ${\mathbb C}$ equipped with a coproduct $\Delta\colon\! {\mathcal Y}(\mathfrak{gl}_{\mathcal N})\! \to{\mathcal Y}(\mathfrak{gl}_{\mathcal N}) \otimes {\mathcal Y}(\mathfrak{gl}_{\mathcal N})$ such that~(i) in Remark~\ref{remy} is satisfied and hence
\begin{equation*}
\Delta\bigl(L^{(m)}_{a,b}\bigr) = \sum_{c=1}^{\mathcal N}\sum_{k=0}^{m} L^{(k)}_{c,b} \otimes L^{(m-k)}_{a,c}.
\end{equation*}
For instance, for the first couple of generators of the Yangian the coproducts are given for $a,b \in \{1,2, \dots, {\mathcal N}\}$ as
\begin{gather*}
\Delta(L_{a,b}^{(1)}) = L_{a,b}^{(1)} \otimes 1_{\mathcal Y} + 1_{\mathcal Y}\otimes
L_{a,b}^{(1)},\\
\Delta(L_{a,b}^{(2)}) = L_{a,b}^{(2)}
\otimes 1_{\mathcal Y}+ 1_{\mathcal Y}\otimes L_{a,b}^{(2)} +
\sum _{c=1}^{\mathcal N}
L_{c,b}^{(1)}\otimes L_{a,c}^{(1)}, \\
\Delta(L_{a,b}^{(3)}) = L_{a,b}^{(3)}
\otimes 1_{\mathcal Y}+ 1_{\mathcal Y}\otimes L_{a,b}^{(3)} +
\sum _{c=1}^{\mathcal N}
L_{c,b}^{(1)}\otimes L_{a,c}^{(2)} +
\sum _{c=1}^{\mathcal N}
L_{c,b}^{(2)}\otimes L_{a,c}^{(1)}, \qquad \dots \,.
\end{gather*}
The coproduct is {coassociative}, and the $n$-coproducts can be derived by iteration via
\[\Delta^{(n+1)} = \bigl(\id \otimes \Delta^{(n)}\bigr)\Delta = \bigl(\Delta^{(n)} \otimes \id\bigr)\Delta.\]

Moreover, the counit exists $\epsilon\colon {\mathcal Y}(\mathfrak{gl}_{\mathcal N}) \to {\mathbb C}$ such that
\[
(\epsilon \otimes \id)\Delta\bigl(L_{a,b}^{(m)}\bigr) = (\id \otimes \epsilon)\Delta\bigl(L^{(m)}_{a,b}\bigr) = L_{a,b}^{(m)},
\]
and hence we obtain by iteration that \smash{$\epsilon\bigl(L^{(m)}_{a,b}\bigr) =0$}, $a,b \in \{1,2,\dots, {\mathcal N}\}$, and $m= {\mathbb Z}^+$.
The antipode $s\colon {\mathcal Y}(\mathfrak{gl}_{\mathcal N})\to {\mathcal Y}(\mathfrak{gl}_{\mathcal N})$ exists, such that
\[
m(s\otimes \id)\Delta\bigl(L_{a,b}^{(m)}\bigr) =m(\id \otimes s)\Delta\bigl(L_{a,b}^{(m)}\bigr)=\epsilon\bigl(L^{(m)}_{a,b}\bigr)1_{\mathcal Y}
\]
and recalling that \smash{$\epsilon\bigl(L_{a,b}^{(m)}\bigr) =0$}, we obtain
\begin{equation*}
\sum_{c=1}^{\mathcal N}\sum_{k=0}^{m} s\bigl(L^{(k)}_{c,b}\bigr) L^{(m-k)}_{a,c} =\sum_{c=1}^{\mathcal N}\sum_{k=0}^{m} L^{(k)}_{c,b} s\bigl(L^{(m-k)}_{a,c}\bigr) =0.
\end{equation*}
For example, the antipode for the first couple of generators is given as
\begin{gather*}
 \bigl(L^{(1)}_{a,b}\bigr) = - L_{a,b}^{(1)}, \nonumber\\
 s\bigl(L^{(2)}_{a,b}\bigr) =
 -L^{(2)}_{a,b} +\sum_{c=1}^{\mathcal N}L^{(1)}_{c,b}L^{(1)}_{a,c},\nonumber\\
 s\bigl(L^{(3)}_{a,b}\bigr) =
 -L^{(3)}_{a,b} +\sum_{c=1}^{\mathcal N}L^{(1)}_{c,b}L^{(2)}_{a,c} + \sum_{c=1}^{\mathcal N}L^{(2)}_{c,b}L^{(1)}_{a,c} -\sum_{c,d=1}^{\mathcal N} L^{(1)}_{d,b}L^{(1)}_{c,d}L^{(1)}_{a,c}, \qquad \dots\,. \end{gather*}

Also, for the first couple of the alternative generators of the Yangian algebra, we obtain (see also~\cite{Chari})
\begin{gather*}
\Delta\bigl({\mathrm Q}_{a,b}^{(1)}\bigr) = {\mathrm Q}_{a,b}^{(1)} \otimes 1_{\mathcal Y} + 1_{\mathcal Y}\otimes
{\mathrm Q}_{a,b}^{(1)} \cr
\Delta\bigl({\mathrm Q}_{a,b}^{(2)}\bigr) = {\mathrm Q}_{a,b}^{(2)}
\otimes 1_{\mathcal Y}+ 1_{\mathcal Y}\otimes {\mathrm Q}_{a,b}^{(2)} + \frac{1}{2}
\sum _{d=1}^{\mathcal N}\bigl({\mathrm Q}_{a,d}^{(1)}\otimes
{\mathrm Q}_{d,b}^{(1)}-
{\mathrm Q}_{d,b}^{(1)}\otimes {\mathrm Q}_{a,d}^{(1)}\bigr), 
\end{gather*}
and \smash{$\epsilon\bigl({\mathrm Q}^{(1)}_{a,b}\bigr) = \epsilon\bigl({\mathrm Q}^{(2)}_{a,b}\bigr) =0$}, and
\smash{$s\bigl({\mathrm Q}^{(1)}_{a,b}\bigr) = - {\mathrm Q}_{a,b}^{(1)}$}, \smash{$s\bigl({\mathrm Q}^{(2)}_{a,b}\bigr) =
 - {\mathrm Q}^{(2)}_{a,b} + \frac{1}{2} {\mathrm Q}^{(1)}_{a,b}$}.

The FRT construction typically provides the coproducts of the associated quantum algebra and thus also leads to the derivation of the counit and antipode. Recall the well-known statement, that given the coproduct of a Hopf algebra one can uniquely determine the counit and antipode of the Hopf algebra from the axioms of the Hopf algebra~\cite{Chari, Majid}.
In the next subsection, we examine tensor realizations of the Yangian and express them as elements of the discrete Magnus expansion.
From the tensor realizations, we derive the $N$-coproducts of the algebra generators and express them in terms of suitable
tridendriform and pre-Lie algebra operations.

\subsection{Tensor realizations of the Yangian and pre-Lie algebras}

 In order to demonstrate the links between the Yangian,
tridendriform and pre-Lie algebras, we consider tensor realizations of the Yangian.

Let us first recall the tensor realizations of any quantum algebra ${\mathfrak A}$, Definition~\ref{frt}. Let
\smash{${\mathbb L}(\lambda) =\sum_{m\geq 0}\frac{L^{(m)}}{\lambda^m} \in \EEnd\bigl({\mathbb C}^{\mathcal N}\bigr) \otimes
{\mathfrak A}$} satisfy equation~(\ref{RTT}) with \smash{$R \in \EEnd\bigl(\bigl({\mathbb C}^{\mathcal N}\bigr)^{\otimes 2}\bigr)$} being a~solution of the Yang--Baxter equation and \smash{$L^{(m)} = \sum_{x, y=1}^{{\mathcal N}} e_{x, y} \otimes L^{(m)}_{x,y}\in \EEnd\bigl({\mathbb C}^{\mathcal N}\bigr) \otimes
{\mathfrak A}$}, where~\smash{$L^{(m)}_{x,y} $} are the generators of the algebra $\mathfrak{A}$.
We recall the quantum monodromy matrix {(or simply monodromy)}
\smash{$T\in \EEnd\bigl({\mathbb C}^{\mathcal N}\bigr) \otimes
\mathfrak{A}^{\otimes N}$}:
\begin{equation}
T_{0,12\dots N}(\lambda) := \bigl(\id \otimes \Delta^{(N)}\bigr){\mathbb L}(\lambda)=
{\mathbb L}_{0N}(\lambda) \cdots {\mathbb L}_{01}(\lambda),\label{T}
\end{equation}
which also satisfies the algebraic relation~(\ref{RTT}) (i.e., provides a tensor realization of the algebra defined by~(\ref{RTT})) and is also a solution of the discrete evolution problem~(\ref{discrete}).
It is shown that the monodromy matrix $T$ satisfies~(\ref{RTT}) by repeatedly using equation~(\ref{RTT}).
Historically, the index $0$ is called ``auxiliary'', whereas the indices $1,2, \dots, N$ are called ``quantum'',
and they are usually suppressed for simplicity, i.e., we simply write $T_{0, N+1}$ (or $T_{N+1}$).

\begin{rem}{\label{tt}} We define the transfer matrix ${\mathfrak t}_{N+1}(\lambda) := \operatorname{tr}_0 (T_{0,N+1}(\lambda)) \in
{\mathfrak A}^{\otimes N}$.
Recall that the monodromy matrix $T$ satisfies~(\ref{RTT}), and hence it is shown that the transfer matrix provides a~family of mutually
commuting quantities~\cite{FadTakRes}
\begin{gather*}
{\mathfrak t}_{N+1}(\lambda) =\lambda^N \sum_{k=1}^N \frac{{\mathfrak t}_{N+1}^{(k)}}{\lambda^k},\\
[ {\mathfrak t}_{N+1}(\lambda), {\mathfrak t}_{N+1}(\mu)] =0\qquad \forall\ \lambda, \mu \in {\mathbb C}\quad \Rightarrow\quad \bigl[ {\mathfrak t}_{N+1}^{(k)}, {\mathfrak t}_{N+1}^{(l)}\bigr] =0.
\end{gather*}
These commutation relations guarantee the ``quantum integrability'' of a spin-chain like system with periodic boundary conditions. For instance, the Hamiltonian and momentum of the system belong to the family of the mutual commuting quantities.
\end{rem}

We focus henceforth on the Yangian, with the $R$-matrix being $R(\lambda) = 1_{V \otimes V} + \lambda^{-1}{\mathcal P}$ (recall~${\mathcal P}$ is the permutation operator),
\[
{\mathbb L}(\lambda) =1_{V \otimes {\mathcal Y}} + \sum_{m>0}\frac{L^{(m)}}{\lambda^m} \in \EEnd\bigl({\mathbb C}^{\mathcal N}\bigr) \otimes
{\mathcal Y}(\mathfrak{gl}_{\mathcal N}), \qquad L^{(m)} = \sum_{x, y=1}^{{\mathcal N}} e_{x, y} \otimes L^{(m)}_{x,y},\]
where
\smash{$L^{(m)}_{x,y} $} are the generators of the Yangian $\mathfrak{gl}_{\mathcal N}$. We also note that a standard simple solution of the fundamental relation~(\ref{RTT}) for the Yangian is ${\mathbb L}(\lambda) = 1 +\lambda^{-1} {\mathbb P}$, where \smash{${\mathbb P} = \sum_{i,j =1}^{\mathcal N} e_{i,j} \otimes {\mathbb P}_{i,j}$} and ${\mathbb P}_{i,j} \in \mathfrak{gl}_{\mathcal N}:$ \smash{$[ {\mathbb P}_{i,j}, {\mathbb P}_{k,l}] = \delta_{i,l} {\mathbb P}_{k,j}-\delta_{k,j} {\mathbb P}_{i,l}$}.

\begin{pro} Let \smash{$L_{a,b}^{(m)}$}, $a,b \in \{1, 2,\dots, {\mathcal N } \}$, $m
\in\{1, 2, \dots\}$, be the generators of the Yangian~$\mathfrak{gl}_{\mathcal N}$~\eqref{fund2b}. Let also the tridendriform algebra operation $\prec$ defined as
$(a\prec b)_n := a_n\Sigma(b)_n $, where \smash{$\Sigma(a)_n= \sum_{m=1}^{n-1} a_m$}, $a, b \in \EEnd\bigl({\mathbb C}^{\mathcal N}\bigr) \otimes{\mathcal Y}$ $($see also Remark~$\ref{class})$. Then, the coproducts of the algebra generators are expressed in terms of the tridendriform algebra operation~$\prec$~as
\begin{align}
\Delta^{(N)}\bigl(L^{(m)}_{a,b}\bigr) &= \sum_{\sum_{j=1}^km_j=m}\sum_{1\leq n_1<\cdots < n_k\leq N} \bigl(\bigl(L^{(m_k)}_{a,b_k}\bigr)_{n_k}\bigl(L^{(m_{k-1})}_{b_k,b_{k-1}}\bigr)_{n_{k-1}} \cdots \bigl(L^{(m_1)}_{b_1, b} \bigr)_{n_1} \bigr)\cr
 &= \sum_{\sum_{j=1}^km_j=m} \sum_{n=1}^N \bigl(L^{(m_k)}_{a,b_k} \prec\bigl( L^{(m_{k-1})}_{b_k,b_{k-1}}\prec \bigl( \cdots \prec \bigl( L^{(m_2)}_{b_2, b_1}\prec L^{(m_1)}_{b_1, b} \bigr) \cdots \bigr)\bigr) \bigr)_n. \label{mine}
\end{align}
\end{pro}
\begin{proof}
Before we carry on with the proof, it is useful to recall the tensor index notation: let~$A$,~$B$ be elements of a unital, associative ${\mathbb C}$-algebra
${\mathfrak A}$, then $A_n B_m = B_m A_n$, $n>m$, where the indices~$n$,~$m$ denote the position of the objects~$A$,~$B$ on an $N$-tensor product, i.e.,
\begin{gather}
 A_n := 1_{\mathfrak A} \otimes \dots\otimes 1_{\mathfrak A}\otimes \underbrace{A}_{\text{n$^{\rm th}$ position}} \otimes 1_{\mathfrak A}\otimes \dots \otimes 1_{\mathfrak A}, \label{nota}\\
 A_n B_m :=1_{\mathfrak A} \otimes \dots\otimes 1_{\mathfrak A}\otimes \underbrace{B}_{\text{m$^{\rm th}$ position}} \otimes 1_{\mathfrak A}\otimes \dots\otimes 1_{\mathfrak A} \otimes \underbrace{A}_{\text{n$^{\rm th}$ position}} \otimes 1_{\mathfrak A} \otimes\dots \otimes 1_{\mathfrak A}.\nonumber
\end{gather}

We recall that the quantum monodromy is expressed as
\[
T_{N+1}(\alpha)\! =\! 1 + \sum_{m>0} \alpha^m T^{(m)}(N+1) , \qquad \alpha = \frac{1}{\lambda},
\]
 with coefficients given in~(\ref{tt2}) and in terms of a tridendriform binary operation in~(\ref{dendr2}) and
 \[
 T^{(m)}(N+1) = \sum_{a,b =1}^{{\mathcal N}} e_{a,b} \otimes T_{a,b}^{(m)}(N+1).
 \]
 We also recall that \smash{$T_{N+1}(\lambda)= \bigl(\id \otimes \Delta^{(N)}\bigr){\mathbb L}(\lambda)$}, hence \smash{$T^{(m)}_{a,b}(N+1) = \Delta^{(N)}\bigl(L^{(m)}_{a,b}\bigr)$} ($\Delta$ is an algebra homomorphism), which via~(\ref{T}) leads to~(\ref{mine}).

We read off the coproducts for the first three orders $m=1,2,3$:
\begin{gather*}
(1) \quad \displaystyle\Delta^{(N)}\bigl(L^{(1)}_{a,b}\bigr) = \sum_{n=1}^N \bigl(L^{(1)}_{a,b}\bigr)_n ,\\
 (2) \quad \Delta^{(N)}\bigl(L^{(2)}_{a,b}\bigr) = \sum_{n=1}^N \bigl(L^{(2)}_{a,b}\bigr)_n + \sum_{n=1}^N \bigl( L^{(1)}_{a,b_1} \prec L^{(1)}_{b_1, b} \bigr)_n,\\
(3) \quad \Delta^{(N)}(L^{(3)}_{a,b}) = \sum_{n=1}^N \bigl(L^{(3)}_{a,b}\bigr)_n + \sum_{n=1}^N \bigl( \bigl( L^{(2)}_{a,b_1} \prec L^{(1)}_{b_1, b} \bigr)_n+ \bigl( L^{(1)}_{a,b_1} \prec L^{(2)}_{b_1, b} \bigr)_n\bigr)\\ \hphantom{(3) \quad \Delta^{(N)}(L^{(3)}_{a,b}) =}{} + \sum_{n=1}^N \bigl( L^{(1)}_{a,b_2} \prec \bigl(L^{(1)}_{b_2, b_1} \prec L^{(1)}_{b_1, b} \bigr)\bigr)_n .\tag*{\qed}
\end{gather*}
\renewcommand{\qed}{}
\end{proof}

We note that a presentation of the Yangian is given via the evaluation homomorphism, $\mbox{ev}\colon {\mathcal Y}(\mathfrak{gl}_{\mathcal N}) \to \mathfrak{gl}_{\mathcal N}$ such that \smash{$L_{a,b}^{(m)} \mapsto \theta^m {\mathbb P}_{a,b}$}, $\theta \in {\mathbb C}$, and
${\mathbb P}_{a,b}$ are the $\mathfrak{gl}_{\mathcal N}$ generators.
We also recall that \smash{${\mathbb L}(\lambda) =1 + \sum_{m>0} \lambda^{-m} {\mathbb P}^m$} is a solution of~(\ref{RTT}), hence the following map also exists, \smash{$\sigma\colon {\mathcal Y}(\mathfrak{gl}_{\mathcal N}) \to \mathfrak{gl}_{\mathcal N}$} such that \smash{$L_{a,b}^{(m)} \mapsto ({\mathbb P}^m)_{a,b}$}. Before we come to the second key connection between the coproducts of the alternative Yangian generators ${\mathrm Q}_{a, b}$ and pre-Lie algebras, we first introduce a useful lemma and some handy notation.

\begin{lem}{\label{lq}}
Let $\mathcal{A}$ be a Rota--Baxter algebra and $R\colon \mathcal{A} \to \mathcal{A}$ be Rota--Baxter operator, such that $[R(x), y] = [R(y), x] =0$, for all $x, y \in \mathcal{A}$. Let also the binary operations of the tridendriform algebra be defined in~\eqref{action2}. Then $x\prec y = y\succ x$, for all $x, y \in \mathcal{A}$.
\end{lem}
\begin{proof} From the definitions in~(\ref{action2}), we conclude that $x\prec y = y\succ x$ for all $x,y \in \mathcal{A}$
\end{proof}

\begin{rem}\label{nota2} We introduce some notation that will be used in the following proposition. Let \smash{${\mathcal O} \in \EEnd\bigl({\mathbb C}^{\mathcal N} \bigr) \otimes {\mathfrak A}$}, then \smash{${\mathcal O}_n =\sum_{a, b}^{\mathcal N}e_{a,b} \otimes ({\mathcal O}_{a,b})_n$}, where $({\mathcal O}_{a,b})_n $ is defined as in~(\ref{nota}). We assume that $A, B \in \EEnd\bigl({\mathbb C}^{\mathcal N} \bigr) \otimes {\mathfrak A}$, and we define
\begin{align*}
\bigl((A \triangleright B)_n\bigr)_{a,b} :={}& \Biggl(\Biggl[ \sum _{m=1}^n A_m, B_n \Biggr] + A_n B_n\Biggr)_{a,b} \\
={}& \sum_{c =1}^{\mathcal N} \Biggl( \sum_{m=1}^{n-1} (A_{a,c})_m (B_{c,b})_n- (B_{a,c})_n\sum_{m=1}^{n-1} (A_{c,b})_m + (A_{a,c})_n (B_{c,b})_n \Biggr). \nonumber
\end{align*}
Recalling the definitions of the tridendriform binary operations~(\ref{action2}), we conclude for the above expression
\begin{equation*}
((A \triangleright B)_n)_{a,b} = \sum_{c =1}^{\mathcal N} ( (A_{a,c} \succ B_{c,b})_n- (B_{a,c} \prec A_{c,b})_n + (A_{a,c} B_{c,b})_n).
\end{equation*}
Due to the fact that $[(A_{a,c})_m, (B_{c,b})_n] =0$, $n\neq m$, Lemma~\ref{lq} applies.
\end{rem}

\begin{pro} Let \smash{${\mathrm Q}_{a,b}^{(m)}$}, $a,b \in \{1, 2,\dots, {\mathcal N } \}$, $m \in {\mathbb Z}^+$, be the alternative generators of the Yangian $\mathfrak{gl}_{\mathcal N}$ $($Lemma~$\ref{lemmaq})$. The coproducts of these generators are emerging from a pre-Lie algebra and are explicitly expressed in terms of a tridendriform algebras binary operations.
\end{pro}
\begin{proof}
First we use the fact that \smash{$\bigl(\id\otimes \Delta^{(N)}\bigr){\mathrm Q}^{(m)} = {\mathcal Q}^{(m)}(N+1)$}~(\ref{T}), we also recall the expressions \smash{${\mathcal Q}^{(m)}(N+1)$}~(\ref{pre1}), then via~(\ref{a0}) and~(\ref{acd}) we conclude
\begin{gather}
 \bigl(\id\otimes \Delta^{(N)}\bigr){\mathrm Q}^{(1)}= \sum_{n=1}^N{\mathrm Q}^{(1)}_n, \cr
 \bigl(\id\otimes \Delta^{(N)}\bigr){\mathrm Q}^{(2)}=-\frac{1}{2} \sum_{n=1}^N
 \bigl({\mathrm Q}^{(1)} \triangleright {\mathrm Q}^{(1)}\bigr)_n +\sum_{n=1}^N\bigl({\mathrm Q}_n^{(2)} +\frac{1}{2} \bigl({\mathrm Q}^{(1)}_n\bigr)^2\bigr),\nonumber\\
 \bigl(\id\otimes \Delta^{(N)}\bigr){\mathrm Q}^{(3)}= \sum_{n=1}^N\biggl(
\frac{1}{4} \bigl(\bigl({\mathrm Q}^{(1)} \triangleright {\mathrm Q}^{(1)}\bigr) \triangleright {\mathrm Q}^{(1)}\bigr)_n +\frac{1}{12} \bigl({\mathrm Q}^{(1)} \triangleright ({\mathrm Q}^{(1)} \triangleright {\mathrm Q}^{(1)})\bigr)_n \biggr)\cr
\hphantom{\bigl(\id\otimes \Delta^{(N)}\bigr){\mathrm Q}^{(3)}= \sum_{n=1}^N\biggl(}{} -
\frac{1}{2} \sum_{n=1}^N\bigl(\bigl({\mathrm Q}^{(2)} \triangleright {\mathrm Q}^{(1)}\bigr)_n+ \bigl({\mathrm Q}^{(1)}\triangleright {\mathrm Q}^{(2)}\bigr)_n\bigr) \cr
\hphantom{\bigl(\id\otimes \Delta^{(N)}\bigr){\mathrm Q}^{(3)}= \sum_{n=1}^N\biggl(}{} -
\frac{1}{4} \sum_{n=1}^N\bigl(\bigl(\bigl({\mathrm Q}^{(1)}\bigr)^2 \triangleright {\mathrm Q}^{(1)}\bigr)_n+ \bigl({\mathrm Q}^{(1)}\triangleright \bigl({\mathrm Q}^{(1)}\bigr)^2 \bigr)_n\bigr) \nonumber\\
\hphantom{\bigl(\id\otimes \Delta^{(N)}\bigr){\mathrm Q}^{(3)}= \sum_{n=1}^N\biggl(}{}
+ \sum_{n=1}^N \biggl({\mathrm Q}_n^{(3)} + \frac{1}{2} \bigl( {\mathrm Q}_n^{(2)} {\mathrm Q}_n^{(1)} +{\mathrm Q}_n^{(1)} {\mathrm Q}_n^{(2)} \bigr) + \frac{1}{6} \bigl({\mathrm Q}^{(1)}_n\bigr)^3\biggr), \quad \dots\,. \!\!\label{pre1b}
\end{gather}

The $N$-coproducts of the alternative generators \smash{${\mathrm Q}^{(m)}_{a,b}$} of the Yangian are extracted from expressions~(\ref{pre1b}), by recalling
\begin{gather*}
\bigl(\id \otimes \Delta^{(N)}\bigr){\mathrm Q}^{(n)}= \sum_{a, b} e_{a,b} \otimes \Delta^{(N)}\bigl({\mathrm Q}_{a,b}^{(n)}\bigr) \qquad {\mathcal Q}_{a,b}^{(p)}(N) = \Delta^{(N)}\bigl({\mathrm Q}_{a,b}^{(p)}\bigr),\qquad\!
a,b \in \{1, \dots, {\mathcal N}\}.
\end{gather*}
Specifically, for the first couple of terms, $n=1, 2$, we obtain from~(\ref{pre1b}), Remark~\ref{nota2},
\begin{gather*}
(1)\quad \Delta^{(N)}\bigl({\mathrm Q}^{(1)}_{a,b}\bigr) = \sum_{n=1}^N\bigl({\mathrm Q}^{(1)}_{a,b}\bigr)_n,\\
(2) \quad \Delta^{(N)}\bigl({\mathrm Q}^{(2)}_{a,b}\bigr) = \sum_{n=1}^N \bigl( \bigl( {\mathrm Q}^{(2)}_{a,b}\bigr)_n -
\frac{1}{2} \sum_{c =1}^{\mathcal N} \bigl( \bigl({\mathrm Q}^{(1)}_{a,c} \succ {\mathrm Q}^{(1)}_{c,b}\bigr)_n- \bigl({\mathrm Q}^{(1)}_{a,c} \prec {\mathrm Q}^{(1)}_{c,b}\bigr)_n \bigr)\bigr).
\end{gather*}
The notation and the requirements introduced in Remark~\ref{nota2} apply in this case.
\end{proof}

We present below an example/application associated with the alternative discrete expansion described in Section~\ref{section3.1}. More details on the example will be presented in future studies, as this is of particular interest.
\begin{exam}{\label{e1}} Let $R\colon {\mathbb C}^{\mathcal N} \otimes {\mathbb C}^{\mathcal N} \to {\mathbb C}^{\mathcal N} \otimes {\mathbb C}^{\mathcal N} $ be a solution of the YBE of the form $R(\lambda) = r + \lambda^{-1} {\mathcal P}$, where $r$ is also a solution of the YBE and $ {\mathcal P}$ is the permutation operator (see, e.g.,~\cite{Doikoutw, DoiSmo2, DoiSmo1}). Such solutions can be obtained for instance from involutive set-theoretic solutions of the YBE~\cite{DoiSmo1, Drin}. Recall now the expression of the monodromy matrix $T(\lambda) = {\mathbb L}_N(\lambda) \cdots {\mathbb L}_1(\lambda)$, where in this example ${\mathbb L}(\lambda) \to R(\lambda)$. Given that $R$ satisfies the YBE, one shows that the monodromy matrix naturally satisfies the~(\ref{RTT}). In this case, expressions~(\ref{alt}) and~(\ref{alt0}), (1)--(3) in Section~\ref{section3.1} hold: $M \to r$, $L \to {\mathcal P}$ and $\alpha = \lambda^{-1}$.
\end{exam}

\subsubsection*{Classical integrability}
Connections between (tri)dendriform, pre-Lie algebras and classical integrable systems and the associated
classical deformed algebras are naturally identified given the findings of the previous section (see also work related to integrable ODEs~\cite{Gol}).
The key point in the description and construction of classical integrable systems from the Hamiltonian point of view is the existence of a classical matrix $r \in \EEnd\bigl({\mathbb C}^{\mathcal N} \otimes {\mathbb C}^{\mathcal N}\bigr)$ that satisfies the classical YBE~\cite{FT, TST} (see also~\cite{RotYB})
\begin{equation*}
[r_{12}(\lambda_1-\lambda_2), r_{13}(\lambda_1-\lambda_3)] + [r_{12}(\lambda_1-\lambda_2)+r_{13}(\lambda_1-\lambda_3), r_{23}(\lambda_2-\lambda_3)] =0.
\end{equation*}
The classical YBE can be seen as a linearization of the quantum YBE, i.e., we set
\[
R= 1 + \delta r + {\mathcal O}\bigl(\delta^2\bigr).
\]

The classical Lax operator
\[
{\mathbb L}_n(\lambda) = \sum_{a,b =1}^{\mathcal N} e_{a,b} (L_{a,b}(\lambda))_n\in \EEnd\bigl({\mathbb C}^{\mathcal N}\bigr),
\]
 where recall $(L_{a,b}(\lambda))_n = \smash{\sum_{k \geq 0} \lambda^{-k}\bigl(L^{(k)}_{a,b}\bigr)_n}$ (see also Remark~\ref{class}) is associated with a discrete space integrable system and satisfies~\cite{FT, Skl1, Skl2}
\begin{equation}
\{ {\mathbb L}_n(\lambda_1) \underset{,}{\otimes} {\mathbb L}_m(\lambda_2)\} = [r(\lambda_1-\lambda_2), { \mathbb L}_n(\lambda_1)\otimes {\mathbb L}_m(\lambda_2)]\delta_{n,m}. \label{sklyanin}
\end{equation}
A Poisson bracket of the form~(\ref{sklyanin}) is called a {\it Sklyanin bracket}. For any two ${\mathcal N}\times {\mathcal N}$ matrices \smash{$A =\sum_{i,j=1}^{\mathcal N} e_{i,j}A_{i,j}$} and \smash{$B = \sum_{i,j=1}^{\mathcal N}e_{i,j} B_{i,j}$}, the bracket \smash{$\{A \underset{,}{\otimes} B\}$} is determined by the brackets of its matrix elements (see also, for example,~\cite{Chari, FT})
\begin{equation}
\{A \underset{,}{\otimes} B\} = \sum_{i,j, k,l=1}^{\mathcal N} e_{ij}\otimes e_{k,l}\ \{A_{i,j}, B_{k,l}\}, \nonumber
\end{equation}
i.e., \smash{$\{A \underset{,}{\otimes} B\}_{ij, kl} =\{A_{i,j}, B_{k,l}\}$}.
Equation~(\ref{sklyanin}) is the classical analogue of the fundamental relation~(\ref{RTT}),
that is, the matrix elements \smash{$L^{(k)}_{a,b}$} are the generators of a Poisson algebra defined
by~(\ref{sklyanin}) (the classical analogue of a quantum algebra).

As in the quantum case, algebraic quantities in involution exist due the existence of an $r$-matrix. Indeed, the monodromy matrix $T(\lambda) = {\mathbb L}_N(\lambda) \cdots {\mathbb L}_1(\lambda) \in \EEnd\bigl({\mathbb C}^{\mathcal N}\bigr)$ also satisfies Sklyanin's bracket and hence it is shown that ${\mathfrak t}(\lambda)= \operatorname{tr} T(\lambda)$ satisfies $\{ {\mathfrak t}(\lambda), {\mathfrak t}(\mu)\} =0 $, for all~${\lambda, \mu \in {\mathbb C}}$, which, given that \smash{${\mathfrak t}(\lambda) = \sum_{m=1}^N\lambda^{-m}{\mathfrak t}^{(m)}$}, leads to a family of quantities in involution \smash{$\big\{ {\mathfrak t}^{(m)}, {\mathfrak t}^{(k)}\big\} =0 $}, i.e., classical integrability {\`a} la Liouville is shown~\cite{FT}.

In the case of the classical $ \mathfrak{gl}_{\mathcal N}$ Yangian, the $r$-matrix is given as \smash{$r(\lambda) = \frac{{\mathcal P}}{\lambda}$}, where we recall that ${\mathcal P}$ is the \smash{${\mathcal N}^2 \times {\mathcal N}^2$} permutation operator \big(recall ${\mathcal P} (a\otimes b) = b \otimes a$, $a, b \in {\mathbb C}^{\mathcal N}$\big). Also we recall, as in the quantum case, \smash{${\mathbb L}(\lambda) = 1 +\sum_{m>0} \lambda^{-m} L^{(m)}$} and \smash{$L^{(m)} = \sum_{a, b=1}^{\mathcal N } L_{a,b}e_{a,b}$}. By substituting~${\mathbb L}$ in~(\ref{sklyanin}), the classical Yangian relations are recovered, i.e., the classical analogues of~(\ref{fund2b}). Everything holds as in the quantum case, but \smash{$[\ ,\, ] \mapsto - \{\ ,\, \}$}, i.e., \smash{$ L^{(m)}_{a,b}$} are commutative objects, and are the generators of the classical Yangian. \smash{$T_{a,b}^{(m)}$} (called also non-local charges) are the classical analogues of the coproducts \smash{$\Delta^{(N)}l\bigl(L_{a,b}^{(m)}r\bigr)$}.

In the continuum case, the Lax operator \smash{${\mathbb A}(x, \lambda) \in \EEnd\bigl({\mathbb C}^{\mathcal N}\bigr)$} satisfies a linear Sklyanin bracket (see \cite{FT})
\begin{equation*}
\{ {\mathbb A}(x,\lambda_1) \underset{,}{\otimes} {\mathbb A}(y, \lambda_2)\} = [r(\lambda_1-\lambda_2), {\mathbb A}(x,\lambda_1)\otimes 1 + 1 \otimes {\mathbb A}(y,\lambda_2)]\delta(x-y). 
\end{equation*}
In the case of the Yangian, specifically \smash{$r(\lambda) = \frac{{\mathcal P}}{\lambda}$} and \smash{${\mathbb A}(x,
\lambda) = \frac{1}{\lambda} A(x)$}. The monodromy matrix in this case is defined as
\[
{\rm T}(x, \alpha) = \overset{\curvearrowleft} {\mathrm P} \exp \biggl(\int_{0}^{x}{\mathbb A}(\xi, \alpha) {\rm d}\xi\biggr)
\] (see also~(\ref{eq:CTS_T}) and Remark~\ref{remcont}) and satisfies the quadratic relation~(\ref{sklyanin})~\cite{FT, Mac}. Also, the expressions for the continuous Magnus expansion~(\ref{qb2a}),~(\ref{iter}) and~(\ref{qb2b}) hold \big($\alpha = \lambda^{-1}$\big).

With this, we conclude our analysis on the links between quantum and classical algebras, arising in the context of integrable systems, and pre-Lie and (tri)dendriform algebras. Further study of the quantum algebras associated with the Example~\ref{e1} will follow in future investigations. Example~\ref{e1} is of special interest given recent findings on the characterization of the quantum algebra associated with involutive set-theoretic solutions of the YBE as quasi-bialgebras~\cite{Doikoutw, DoGhVl}. These quasi-bialgebras naturally emerge after suitably twisting the Yangian~\cite{Doikoutw}.
We note that set-theoretic solutions do not have a classical analogue, a fact that makes Example~\ref{e1} even more intriguing.
Moreover, it is known that all involutive set-theoretic solutions of the YBE come from braces~\cite{Jesp, Rump1}. This, together with the fact that braces are obtained from pre-Lie algebras, while {pre-Lie} algebras in turn are naturally connected to quantum algebras via the FRT construction, provide strong motivations to further investigate these algebraic structures at the level of solutions, but also at the level of the emerging quantum algebras.

\subsection*{Acknowledgments}
I would like to thank the anonymous referees for their constructive comments and suggestions. Support from the EPSRC research grant EP/V008129/1 is acknowledged.

\pdfbookmark[1]{References}{ref}
\LastPageEnding

\end{document}